\documentclass[journal,twoside]{IEEEtran}

\usepackage{amssymb,amsmath,amsthm} 
\usepackage{stmaryrd} 
\usepackage{dsfont}   
\usepackage{bm}       
\usepackage{cite}     
\usepackage{graphicx} 
\usepackage[caption=false,font=footnotesize]{subfig}  
\usepackage[table]{xcolor}  
\usepackage{yfonts}	

\newtheorem{example}{Example}
\newtheorem{theorem}{Theorem}
\newtheorem{corollary}{Corollary}[theorem]

\DeclareMathOperator{\myre}{Re}
\DeclareMathOperator{\myim}{Im}
\newcommand{\Sin}[1]{\sin\left(#1\right)}
\newcommand{\re}[1]{\myre\left(#1\right)}
\newcommand{\im}[1]{\myim\left(#1\right)}

\newcommand{\pt}{\mathcal{P}_\mathbb{T}}
\newcommand{\lt}{\mathcal{L}_\mathbb{T}}
\newcommand{\ts}{\mathcal{T}_\text{s}}
\newcommand{\krnp}{\mathcal{K}_{\mathcal{D},n_p}}
\newcommand{\range}[2]{\llbracket {#1},{#2} \rrbracket}

\title{Practical Considerations in Direct Detection Under Tukey Signalling}
\author{Amir Tasbihi~\IEEEmembership{Graduate Student
Member, IEEE} and Frank R.
Kschischang~\IEEEmembership{Fellow, IEEE}\thanks{Submitted on March 3rd, 2022
to \emph{J.~Lightw.~Technol.}, revised on October 5th, 2022, accepted on
December 1st, 2022. The authors are with the Edward S. Rogers Sr.\ Dept.\ of
Electrical \& Computer Engineering, University of Toronto, Toronto, ON M5S 3G4,
Canada.  Email:
\texttt{\{tasbihi,frank\}@ece.utoronto.ca.}}}
\begin{document}
\maketitle

\begin{abstract}
The deliberate introduction of controlled intersymbol interference (ISI) in
Tukey signalling enables the recovery of signal amplitude and (in part) signal
phase under direct detection, giving rise to significant data rate improvements
compared to intensity modulation with direct detection (IMDD).  The use of an
integrate-and-dump detector makes precise waveform shaping unnecessary, thereby
equipping the scheme with a high degree of robustness to nonlinear signal
distortions introduced by practical modulators.  Signal sequences drawn from
star quadrature amplitude modulation (SQAM) formats admit an efficient trellis
description that facilitates codebook design and low-complexity near
maximum-likelihood sequence detection in the presence of both shot noise and
thermal noise.  Under the practical (though suboptimal) allocation of a $50$\%
duty cycle between ISI-free and ISI-present signalling segments, at a symbol
rate of $50~$Gbaud and a launch power of $-10~$dBm the Tukey scheme has a
maximum theoretically achievable throughput of $200~$Gb/s with an $(8,4)$-SQAM
constellation, while an IMDD scheme achieves about $145~$Gb/s using PAM-$8$.
Note that the two mentioned constellations have the same number of magnitude
levels and the difference in throughput is resulting from exploiting phase
information under using a complex-valued signal constellation. 
\end{abstract}

\begin{section}{Introduction}\label{sec:introduction}
\IEEEPARstart{T}{ukey} waveforms are time-limited signals that allow
amplitude and (to some extent) phase recovery of complex-valued symbols
under direction detection by the deliberate introduction of inter-symbol
interference (ISI)~\cite{tukey}, in a manner that is reminiscent of
\textit{partial-response signalling} schemes dating back to the early
1960s~\cite{partial1,partial2}.  Under Tukey signalling, ISI is
controlled so that only adjacent symbols interfere, resulting in an
alternation between ISI-free and ISI-present signalling segments at the
receiver, where detection is accomplished by a low-complexity
integrate-and-dump architecture.

Tukey signalling with direct detection is an alternative to
intensity-modulation with direct detection (IMDD), which encodes only
the intensity, \textit{i.e.}, squared magnitude, of the transmitted
waveform by transmitting only positive real symbols.  In contrast to
IMDD, Tukey signalling with direct detection allows information to be
encoded in the phase, not only the intensity, of the transmitted
symbols;  however, this increased capability necessitates implementation
of an in-phase/quadrature (IQ) modulator at the transmitter and the use
of two (rather than one) analog-to-digital (A/D) converters (each
operating at the symbol rate) at the receiver.  The requirement of having
at least two real-valued samples per complex-valued symbol is inevitable
in all schemes that extract phase (in addition to magnitude).

Another class of communication schemes that can exploit phase under
direct detection is based on the so-called Kramers-Kronig (KK) (or
Hilbert transform) relationship between the phase and the logarithm of
the magnitude of a minimum-phase complex-valued
signal~\cite{kk1,kk2,kk3}.  The use of this relationship in
communication systems dates back to the early
1960s~\cite{old_kk1,old_kk2,old_kk3,old_kk4,old_kk5,old_kk6}, where it
was used to extract phase information from an envelope detector.  To
satisfy the minimum-phase constraint, KK-based schemes must add a
non-information-bearing tone to the complex-valued waveform (either at
the transmitter or at the receiver), resulting in a substantial
inefficiency in overall transceiver power consumption.  Furthermore,
because they must perform bandwidth-broadening nonlinear operations, KK
receivers often demand an over-sampling of about three times the minimum
required sampling rate, \textit{i.e.}, six times the symbol rate.
Although there are KK schemes that enable phase recovery without such
oversampling, these require a very high carrier-to-signal power
ratio~\cite{kk3}. 

In this paper, we address a number of practical concerns that arise with
Tukey signalling under direct detection.  In particular, 
\begin{itemize}
\item we consider the influence of the nonlinearity of the IQ modulator
(which was idealized in~\cite{tukey}) and show that the scheme is quite
robust to modulator imperfections;
\item we fix the duty-cycle between ISI-free and ISI-present signalling
segments to $50\%$ (unlike in~\cite{tukey}, where it was allowed to
range to values as low as $10\%$, which leads to unrealistically short
integration intervals in high baud-rate systems);
\item we introduce a low-complexity near maximum-likelihood (ML)
trellis-based decoding algorithm in the presence of both shot noise and
thermal noise encountered at the output of a p-i-n photodiode under
practical parameter settings (replacing the brute-force search decoding
and avalanche photodiode analysis of \cite{tukey});
\item we consider a wider range of constellation sizes compared
with~\cite{tukey}.
\end{itemize}

The rest of the paper is organized as follows.  Sec.~\ref{sec:system_model}
describes the system model for use in the C band where chromatic
dispersion is significant; in particular,
Sec.~\ref{subsec:IQ_modulator} addresses the nonlinearity of the IQ modulator.
Sec.~\ref{sec:trellis} shows how trellis diagrams can be used for codebook
design and decoding.  While typically trellises are associated with a
\emph{given} codebook, in Sec.~\ref{subsec:sld_trellis} trellises are used to
\emph{find} codebooks composed of symbols drawn from the class of signal
constellations described in Sec.~\ref{subsec:sqam}.
Sec.~\ref{subsec:trellis_decoding} provides a branch-metric function that
enables near-ML trellis-based decoding using the Viterbi algorithm.
Sec.~\ref{sec:simulation} provides extensive numerical simulation results,
while Sec.~\ref{sec:compare} compares Tukey signalling under direct detection
with IMDD.  The system model of Sec.~\ref{sec:system_model}
requires link-distance knowledge to precompensate for chromatic dispersion. In
Sec.~\ref{sec:Oband}, we remove this requirement by operating in the O band,
instead of the C band, leaving residual chromatic dispersion
uncompensated and making the transmitter agnostic
to the transmission-link length.
Finally, Sec.~\ref{sec:discussion} provides various concluding remarks.

Throughout this paper, the integers, real numbers, and complex numbers
are denoted, respectively, by $\mathbb{Z}$, $\mathbb{R}$, and
$\mathbb{C}$.  The set of real-valued and complex-valued functions over
$\mathbb{R}$ are denoted as $\mathbb{R}^{\mathbb{R}}$ and
$\mathbb{C}^{\mathbb{R}}$, respectively.  For any $x\in\mathbb{R}$,
$\mathbb{R}^{>x}$ denotes the set of real numbers strictly greater than
$x$.  The sets $\mathbb{R}^{\geq x}$, $\mathbb{Z}^{>x}$, and
$\mathbb{Z}^{\geq x}$ are defined similarly.  For any $a,b
\in\mathbb{Z}$ with $a \leq b$, let $\range{a}{b} =\{a, a+1, \ldots,
b\}$.  The cardinality of a set $\mathcal{S}$ is denoted by
$|\mathcal{S}|$; thus, \textit{e.g.}, $|\range{0}{m-1}|=m$.  For any
$r\in\mathbb{R}$, $\lceil r\rceil$ and $\lfloor r \rfloor$ denote, respectively, the smallest integer greater
than or equal to $r$ and the largest integer less than or equal to $r$.  For a complex number $w\in\mathbb{C}$, $|w|$,
$\arg(w)$, and $w^\ast$ denote the magnitude, phase, and complex
conjugate of $w$, respectively; furthermore, we assume that
$\arg(w)\in[0,2\pi)$ and $i=\sqrt{-1}$.  The Fourier transform of a
complex-valued function $x(t)$ is denoted as $\mathcal{F}[x(t)]$; the
inverse Fourier transform of $X(f)$ is denoted as
$\mathcal{F}^{-1}[X(f)]$.  For a random variable $X$,
$X\sim\mathcal{N}(\mu,\sigma^2)$ denotes that $X$ has a Gaussian
distribution with mean $\mu$ and variance $\sigma^2$.  Finally, for any
$k \in\mathbb{Z}^{>0}$, $\bm{I}_k$ denotes the $k \times k$ identity
matrix, while for any $k$ and $\ell\in\mathbb{Z}^{>0}$,
$\mathds{1}_{k\times\ell}$ denotes the all-one matrix of size
$k\times\ell$.
\end{section}

\begin{section}{System Model}\label{sec:system_model}

In this section, we describe the system model, as shown in
Fig.~\ref{fig:system_model}.  For simplicity of computations and ease of
reading, we assume transmission over a single polarization; however,
with an obvious modification, the scheme can be adapted to work in a
dual-polarized system.

\begin{figure}
\centering
\includegraphics[scale=0.6666]{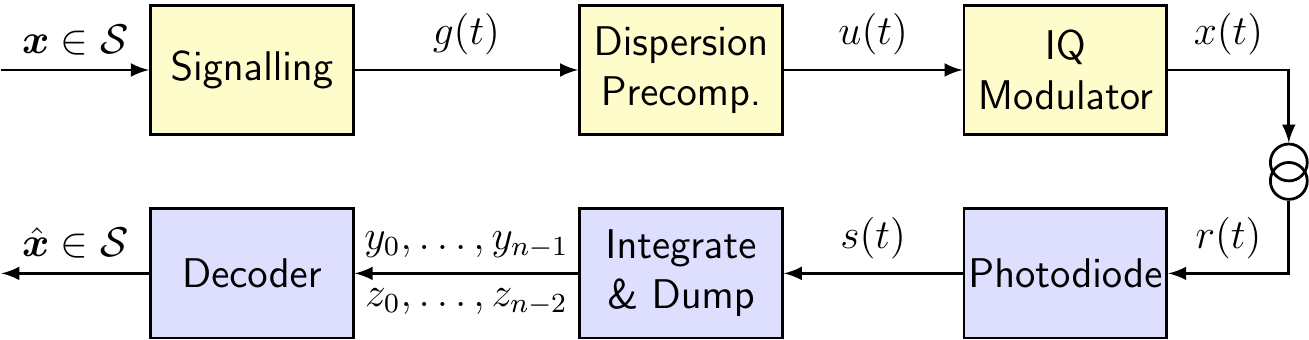}
\caption{The system model for operation in the C band. For the O band system model, see Fig.~\ref{fig:system_model_Oband}.}
\label{fig:system_model}
\end{figure}

\begin{subsection}{The Transmission Medium}
\label{subsec:transmission_medium}

We assume data transmission over a standard single-mode fiber (SSMF) at a
wavelength in the C band, \textit{i.e.}, close to $1550~$nm, where the fiber
has the lowest loss, \textit{i.e.}, about $0.2~\text{dB}\cdot\text{km}^{-1}$.
The lower loss comes at the expense of a higher chromatic dispersion compared
to operating at the zero-dispersion wavelength of SSMF, \textit{i.e.}, about
$1310~$nm in the O band. 

A key application of direct detection is in short-haul data
transmission, \textit{e.g.,} intra data-center communication.
Therefore, power loss and chromatic dispersion are the only fiber
imperfections that we take into account. Other defects of the channel,
\textit{e.g.}, Kerr effect or polarization mode dispersion (for a
dual-polarized transmission), are not considered.  Accordingly, if the
waveform $x(t)\in\mathbb{C}^\mathbb{R}$ is launched into the fiber, the
waveform received at the fiber output is modelled as~\cite[Sec.~4.3.3]{papen2019}
\begin{equation*}
r(t)=\mathcal{F}^{-1}\left[X(f)e^{-\varrho\frac{L}{2}}e^{-i\beta_0L}e^{-2\pi f\beta_1L}e^{-i2\beta_2L\pi^2f^2}\right],
\end{equation*}
where $X(f)=\mathcal{F}[x(t)]$, $L\in\mathbb{R}^{>0}$ is the fiber
length, $\varrho\in\mathbb{R}^{>0}$ is the loss factor and
$\beta_2\in\mathbb{R}$ is the group-velocity dispersion parameter.  At
the wavelength of $1550~$nm, $\varrho=0.046~\text{km}^{-1}$ (equivalent
to a loss of $0.2~\text{dB}\cdot\text{km}^{-1}$) and
$\beta_2=-2.167\times 10^{-23}~\text{s}^2\cdot\text{km}^{-1}$. 
At the operating wavelength $\lambda$, we have 
$\beta_0=2\pi/\lambda$.
Furthermore,
$\beta_1$ is the group delay per unit of length~\cite[Sec.~4.3]{papen2019}.
\end{subsection}

\begin{subsection}{IQ Modulator}
\label{subsec:IQ_modulator}

The IQ modulator comprises two Mach-Zehnder modulators (MZMs) operating
in the push-pull regime. For a modulating waveform
$v(t)\in\mathbb{R}^{\mathbb{R}}$, the output of an MZM in this regime is
$E_\text{out}(t)=E_\text{in}\Sin{\kappa v(t)}$, where $E_\text{in}$ is
an unmodulated continuous-wave narrow-band optical input and
$\kappa\in\mathbb{R}$ is a constant depending on the operating
wavelength and physical properties of the MZM~\cite[Sec.~7.5.3]{papen2019}.  In this paper, we
consider a baseband model; thus, we assume that
$E_\text{in}\in\mathbb{R}^{>0}$.  The IQ modulator uses separate MZMs
for the in-phase and the quadrature components, producing the optical
waveform~\cite[Sec.~7.5.3]{papen2019} 
\begin{equation}
x(t)=E_\text{in}\Big(\Sin{\kappa\re{u(t)}}+i\Sin{\kappa\im{u(t)}}\Big)
\label{eq:driver_signal}
\end{equation}
in response to the complex-valued electrical waveform
$u(t)\in\mathbb{C}^\mathbb{R}$.
\end{subsection}

\begin{subsection}{Signalling}
\label{subsec:signalling}

For any $\beta\in[0,1]$ the Tukey (or time-domain raised
cosine~\cite{munich_tukey}) waveform $w_\beta(t)$ is defined
as~\cite{tukey} 
\[
w_{\beta}(t)\triangleq
\begin{cases}
\alpha_\beta, & \text{if }|t|\leq\frac{(1-\beta)}{2};\\
\frac{\alpha_\beta}{2}\left(1-
\sin\left(\frac{\pi(2|t|-1)}{2\beta}\right)\right),&
\text{if }\left||t|-\frac{1}{2}\right|\leq\frac{\beta}{2};\\
0, & \text{otherwise},
\end{cases}
\]
where $\alpha_\beta\triangleq\frac{2}{\sqrt{4-\beta}}$.  For a
$T\in\mathbb{R}^{>0}$ and a block length $n\in\mathbb{Z}^{>1}$, the
signalling block in Fig.~\ref{fig:system_model} accepts $n$ complex
numbers $x_0,\ldots,x_{n-1}\in\mathbb{C}$ as its input and produces
\begin{equation}
g(t)=\sum_{k\in \range{0}{n-1}} x_kw(t-kT),
\label{eq:g}
\end{equation}
where the signalling waveform $w(t)$ is a scaled, dilated Tukey
waveform; in particular, $w(t)=\frac{1}{\sqrt{T}}w_\beta(\frac{t}{T})$
for some $\beta\in[0,1]$.
\end{subsection}

\begin{subsection}{Dispersion Precompensation}
\label{subsec:dispersion_precompensation}

Chromatic dispersion is precompensated in the electrical domain by a
filter with frequency response $ H(f)=e^{i2\beta_2 L\pi^2f^2}$.  Note
that this linear precompensation does not take into account the
nonlinearity of the IQ modulator and therefore it does not perform ideal
precompensation.  Nevertheless, when the power of the modulating
waveform, $u(t)$, is relatively small, the IQ modulator operates almost
linearly; thus, the dispersion precompensating filter with frequency
response $H(f)$ compensates the chromatic dispersion of the fiber to a
large degree.
\end{subsection}

\begin{subsection}{Photodiode and Noise}
\label{subsec:photodiode}

The output $s(t)\in\mathbb{R}^\mathbb{R}$ of the photodiode, in response
to the input waveform $r(t)\in\mathbb{C}^\mathbb{R}$, is given as
\begin{equation*}
s(t)=\mathfrak{R}\,|r(t)|^2 + |r(t)|n_\text{sh}(t)+n_\text{th}(t), 
\end{equation*}
where $\mathfrak{R}$ is the photodiode responsivity, and where
$n_\text{sh}(\cdot)$ and $n_\text{th}(\cdot)$ are zero-mean white
Gaussian random processes with constant two-sided power spectral
densities $\sigma_\text{sh}^2$ and $\sigma_\text{th}^2$, respectively.
The two terms, $|r(t)|n_\text{sh}(t)$ and $n_\text{th}(t)$, are referred
to as shot noise and thermal noise, respectively.
\end{subsection}

\begin{subsection}{Integrate \& Dump}
\label{subsec:integrate_and_dump}

Note that $g(t)$ depends only on $x_k$, $k\in\range{0}{n-1}$,  whenever
$\left|t-kT\right|\leq\frac{(1-\beta)T}{2}$ and depends only on $x_\ell$
and $x_{\ell+1}$, $\ell\in\range{0}{n-2}$, whenever
$\left|t-(\ell+\frac{1}{2})T\right|\leq\frac{1}{2}\beta T$.  Thus, as
in~\cite{tukey}, we define the $k^\text{th}$ ISI-free interval as
\begin{equation*}
\mathcal{Y}_k\triangleq \{ t \in \mathbb{R} \colon | t - kT | \leq (1-\beta)T/2 \}
\end{equation*}
and the $\ell^\text{th}$ ISI-present interval as 
\begin{equation*}
\mathcal{Z}_\ell\triangleq \{ t \in \mathbb{R} \colon | t - \ell T - T/2 | < \beta T/2 \}.
\end{equation*}
The integrate-and-dump unit accepts $s(t)$ as its input and produces
\begin{equation}
y_k=\int_{\mathcal{Y}_k}s(t)~\text{d}t \quad \text{ and } \quad
z_\ell=\int_{\mathcal{Z}_\ell}s(t)~\text{d}t
\label{eq:ykzl}
\end{equation}
for $k\in\range{0}{n-1}$ and $\ell\in\range{0}{n-2}$. 

As we do not compensate for the nonlinearity of the IQ modulator,
determining the exact distribution of $y_k$ and $z_\ell$ given the
transmitted symbols, $x_0,\ldots,x_{n-1}$, is not straightforward.
Thus, the implementation of a true ML receiver seems intractable.
Instead, we use
\begin{equation}
x(t)\simeq E_\text{in}\kappa u(t),
\label{eq:aprx_driver}
\end{equation}
as an approximation for (\ref{eq:driver_signal}) at the decoder, which
becomes increasingly accurate for low-power modulating waveforms.  Note
that this approximation is applied at the receiver, and not at the
transmitter, to simplify computations. In all simulation results
reported in Sec.~\ref{sec:simulation}, the nonlinearity of the IQ
modulator, as given in (\ref{eq:driver_signal}), has been properly taken
into account when simulating the transmitter.

Using (\ref{eq:aprx_driver}), one may approximate $y_k$ in (\ref{eq:ykzl})
as
\begin{equation}
y_k\simeq\mathfrak{R}\,\overline{y_k}+\sqrt{\overline{y_k}}n_k+m_k,
\label{eq:yk_aprx}
\end{equation}
where 
\begin{equation}
	\overline{y_k}=\left(E_\text{in}\kappa\alpha_\beta e^{-\varrho\frac{L}{2}}|x_k|\right)^2(1-\beta)
\in\mathbb{R}^{\geq 0},
\label{eq:y_bar}
\end{equation}
and where $n_k\sim\mathcal{N}(0,\sigma_\text{sh}^2)$,
$m_k\sim\mathcal{N}(0,(1-\beta)T\sigma_\text{th}^2)$, and
$k\in\range{0}{n-1}$.

Define $\psi:\mathbb{C}^2\rightarrow\mathbb{R}^{\geq 0}$ such that, for all
$(a,b)\in\mathbb{C}^2$,
\begin{equation}
\psi(a,b)=\frac{1}{4}|a+b|^2+\frac{1}{8}|a-b|^2.
                \label{eq:psi}
\end{equation}
Note that $\psi(a,b)$ is a function of the magnitudes of $a$ and $b$ and
the cosine of the their phase difference;  thus it follows that
\begin{equation}
\psi( a e^{i\theta}, b e^{i\theta}) = \psi(a,b),
\label{eq:psirotate}
\end{equation}
for any $\theta \in \mathbb{R}$.  By using (\ref{eq:ykzl}) and
(\ref{eq:aprx_driver}), one may approximate $z_\ell$ as
\begin{equation}
z_\ell\simeq\mathfrak{R}\,\overline{z_\ell}+\sqrt{\overline{z_\ell}}p_\ell+q_\ell,
\label{eq:zl_aprx}
\end{equation}
where
\begin{equation}
	\overline{z_\ell}=\left(\alpha_\beta E_\text{in}\kappa e^{-\varrho\frac{L}{2}}\right)^2
\psi(x_\ell,x_{\ell+1})\beta\in\mathbb{R}^{\geq 0}, 
\label{eq:z_bar}
\end{equation}
and where $p_\ell\sim\mathcal{N}(0,\sigma_\text{sh}^2)$,
$q_\ell\sim\mathcal{N}(0,\beta T\sigma_\text{th}^2)$, and
$\ell\in\range{0}{n-2}$.  Furthermore, given the transmitted symbols
$x_0,\ldots,x_{n-1}$, for any $k,k'\in\range{0}{n-1}$ and
$\ell,\ell'\in\range{0}{n-2}$, $n_k$, $m_{k'}$, $p_\ell$, and
$q_{\ell'}$ are mutually independent random variables. Under the
approximations (\ref{eq:yk_aprx}) and (\ref{eq:zl_aprx}), the
conditional distribution of $y_k$ given $x_k$ is
\begin{equation}
y_k\mid x_k\approx\mathcal{N}
\left(\mathfrak{R}\,\overline{y_k},\overline{y_k}\sigma_\text{sh}^2+(1-\beta)T\sigma_\text{th}^2\right),
\label{eq:y_likelihood}
\end{equation}
and the conditional distribution of $z_\ell$ given $x_{\ell}$ and
$x_{\ell+1}$ is
\begin{equation}
z_\ell\mid x_\ell,x_{\ell+1}\approx\mathcal{N}\left(\mathfrak{R}\,\overline{z_\ell},
\overline{z_\ell}\sigma_\text{sh}^2+\beta T\sigma_\text{th}^2\right),
\label{eq:z_likelihood}
\end{equation}
where ``$\approx$" should be read as ``approximately distributed as.''
From (\ref{eq:y_likelihood}) and (\ref{eq:z_likelihood}), it follows
that the detection performance---\textit{i.e.}, bit error rate (BER) or
achievable information rate---depends on the symbol rate.  This is in
contrast with communication over a classical additive white Gaussian
noise (AWGN) channel, where the detection performance does not depend on
the symbol rate.  In the communication scenarios considered in this paper,
shot noise dominates thermal noise at sufficiently high symbol rates.
\end{subsection}

\begin{subsection}{Codebook}
\label{subsec:codebook}

Similar to \cite{tukey}, we define the function
$\Upsilon:\mathbb{C}^n\rightarrow\mathbb{R}^{2n-1}$ as 
\begin{multline*}
\Upsilon(x_0,\ldots,x_{n-1})=\Big(|x_0|^2,\psi(x_0,x_1),|x_1|^2,\psi(x_1,x_2),\\
\ldots, |x_{n-2}|^2,\psi(x_{n-2},x_{n-1}),|x_{n-1}|^2\Big),
\end{multline*}
and, for any $ \bm{x}\in\mathbb{C}^n$, we refer to $\Upsilon(\bm{x})$ as
the \emph{signature} of $\bm{x}$.  From (\ref{eq:yk_aprx}) and
(\ref{eq:zl_aprx}) one may see that any two transmitted complex-valued
$n$-vectors having the same signature are indistinguishable by our
receiver, even in the absence of noise.

We may define an equivalence relation on $\mathbb{C}^n$, deeming two
vectors, $\bm{x}$ and $\bm{\tilde{x}}\in\mathbb{C}^n$ as
\textit{square-law equivalent}, denoted $\bm{x}\equiv\bm{\tilde{x}}$, if
$\Upsilon(\bm{x})=\Upsilon(\bm{\tilde{x}})$.  When $\bm{x}$ and
$\bm{\tilde{x}}$ are not square-law equivalent, they are called
\textit{square-law distinct} (SLD).  We refer to any set $\mathcal{S}
\subset \mathbb{C}^n$ of pairwise SLD $n$-tuples as a \emph{codebook},
and we will always assume that the transmitted $n$-tuples,
$(x_0,\ldots,x_{n-1})$, in (\ref{eq:g}) are chosen from some fixed
codebook.
\end{subsection}

\begin{subsection}{Decoder}
\label{subsec:decoder}

The decoder block in Fig.~\ref{fig:system_model} accepts $y_0,\ldots,y_{n-1}$
and $z_0,\ldots,z_{n-2}$ and produces the detected $n$-tuple $\hat{\bm{x}}$, an
element of the codebook $\mathcal{S}$.  Note that the ISI-present intervals at
the beginning and at the end of the received waveform, which overlap,
respectively, with the previous and the next blocks, are ignored in detection.
Therefore, no time-guard is needed between consecutive blocks. 

From (\ref{eq:yk_aprx}) and (\ref{eq:zl_aprx}) one may conclude that,
independent of the decoder block, any pair of codebooks $\mathcal{S}_1$
and $\mathcal{S}_2$ where
$\Upsilon(\mathcal{S}_1)=\Upsilon(\mathcal{S}_2)$, \textit{i.e.}, where
the codewords produce the same set of signatures, are equivalent from
the perspective of the detector and thus have the same error rate.

In \cite{tukey}, the decoder performs ML detection by a brute-force
search over all $|\mathcal{S}|$ elements of the codebook.  We describe a
more practical lower-complexity trellis-based decoding algorithm in the
next section.
\end{subsection}

\end{section}


\begin{section}{The Trellis Diagram and Star QAM Constellations}
\label{sec:trellis}

In this section we show how trellis diagrams that result from the use of
star quadrature amplitude modulation (SQAM) can be used to determine SLD
sequences for codebook design and for decoding.

\begin{subsection}{Trellises}

Recall that a trellis
$\mathbb{T}=(\mathcal{V},\mathcal{A},\mathcal{E})$, is an edge-labelled
directed graph having a vertex set $\mathcal{V}$, an edge-label set
$\mathcal{A}$, and a set $\mathcal{E} \subset \mathcal{V} \times
\mathcal{A} \times \mathcal{V}$ of labelled edges,
satisfying~\cite{trellis1}:
\begin{enumerate}
\item $\mathcal{V}$ contains distinct vertices $v_r$ (the \emph{root},
having zero in-degree) and $v_g$ (the \emph{goal}, having zero
out-degree);
\item every vertex in  $\mathcal{V}$ is reachable by a directed path
from $v_r$;
\item $v_g$ is reachable by a directed path from every vertex in
$\mathcal{V}$;
\item every directed path from $v_r$ to $v_g$ has the same length,
called the length of $\mathbb{T}$ and denoted by $|\mathbb{T}|$. 	
\end{enumerate}
These properties imply that every directed path from $v_r$ to any fixed
vertex $v \in \mathcal{V}$ has the same length, $d(v)$, called the
\emph{depth} of $v$.  Thus $d(v_r)=0$ and $d(v_g) = |\mathbb{T}|$.  The
vertex set $\mathcal{V}$ can then be partitioned as $\mathcal{V}_0 \cup
\cdots \cup \mathcal{V}_{|\mathbb{T}|}$, where $\mathcal{V}_j$ denotes
the set of vertices having depth $j$.  The subgraph of $\mathbb{T}$
induced by edges incident from $\mathcal{V}_j$ is called the $j$th
\emph{trellis section} of $\mathbb{T}$.

An edge $e = (v_1,a,v_2) \in \mathcal{E}$ (directed from $v_1$ to $v_2$
with label $a$) is said to have label $\lambda(e) = a$.  If
\[
\bm{p} = ((v_1,a_1,v_2),(v_2,a_2,v_3),\ldots,(v_{m-1},a_{m-1},v_m))
\]
is a directed path from a vertex $v_1$ to a vertex $v_m$ in
$\mathbb{T}$, we denote by $\lambda(\bm{p})$ the label sequence
$(a_1,\ldots,a_{m-1})$ associated with that path.  The set of all
directed paths from $v_r$ to $v_g$ in $\mathbb{T}$ is denoted as
$\mathcal{P}_{\mathbb{T}}$, and the set of associated label sequences is
denoted as
\[
\lt \triangleq \{ \lambda(\bm{p}) \mid \bm{p}
\in \mathcal{P}_{\mathbb{T}} \}.
\]
In effect, a trellis $\mathbb{T}$ is simply a convenient graphical
representation of its set of label sequences, $\lt$.
\end{subsection}

\begin{subsection}{Star QAM Constellation}
\label{subsec:sqam}

In this paper we consider a class of signal constellations called
star-quadrature-amplitude modulation (SQAM)~\cite{munich_tukey} having
particularly tractable trellis representations which enable their
investigation.  Star QAM constellations have points in the complex plane
that lie at the intersections of a number of uniformly spaced rays and a
number of concentric rings, as illustrated in Fig.~\ref{fig:sqam}.

More formally, the number of rings is denoted as $n_r$. The ring radii
are determined by a \emph{radius set} $\mathcal{D} =  \{ d_1, d_2,
\ldots, d_{n_r} \}$ of $n_r$ distinct positive real numbers.  The number
of phase angles (or rays) is denoted as $n_p$. The phase angles are
drawn from the set
\[
\Phi_{n_p}\triangleq\left\{0,\frac{2\pi}{n_p},\ldots, \frac{2\pi(n_p-1)}{n_p}\right\},
\]
and the $(n_r,n_p)$-SQAM constellation with radius set $\mathcal{D}$ and
$n_p$ phases is given as
\[
\krnp=  \{ d e^{i \theta} \colon d \in \mathcal{D}, \theta \in \Phi_{n_p} \}.
\]
For later convenience, let
\[
\mathcal{D}_j \triangleq \{ d_j e^{i \theta} \colon \theta \in \Phi_{n_p} \}
\]
be the set of constellation points at radius $d_j$.  In~\cite{tukey},
$\krnp$ is referred to as an \textit{$n_r$-ring/$n_p$-ary phase
constellation}.  For brevity, when $\mathcal{D}$ is clear from context
or irrelevant, we may refer to an $(n_r,n_p)$-SQAM constellation without
mentioning the radius set $\mathcal{D}$ explicitly. 

\newsavebox{\tempbox}
\newsavebox{\tempboxBig}
\begin{figure}
\centering
\sbox{\tempbox}{\includegraphics[scale=0.38]{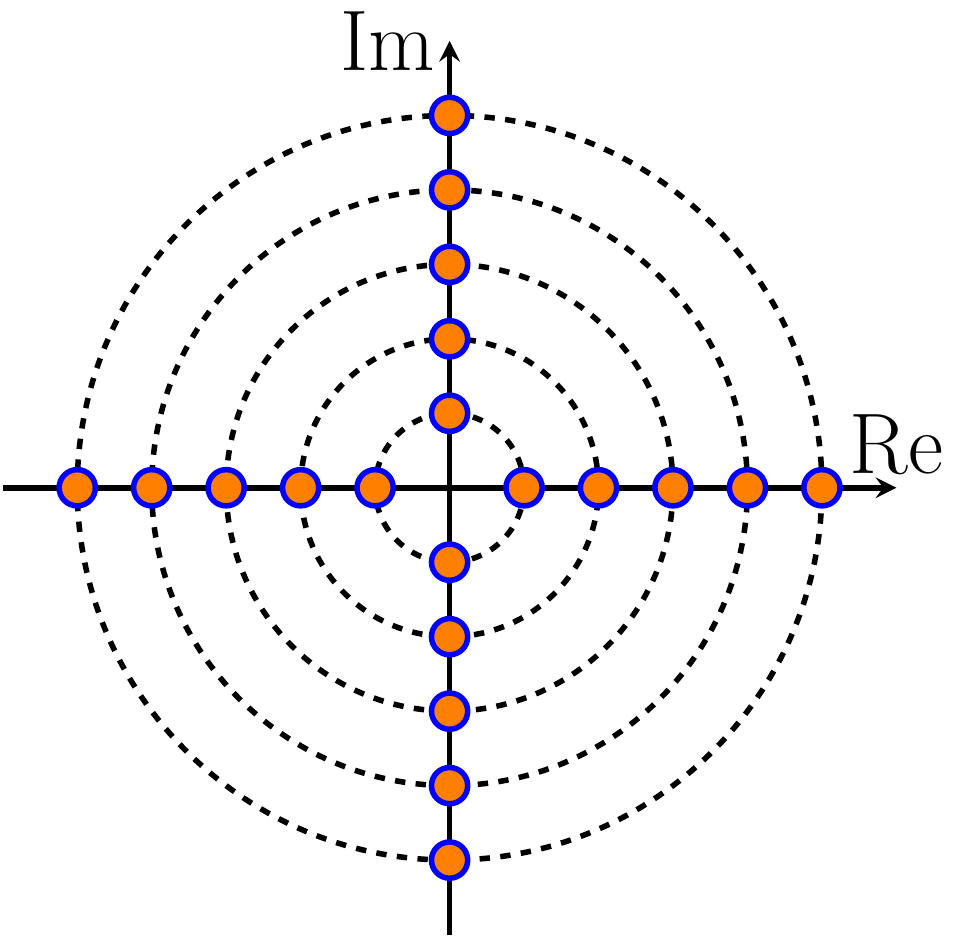}}
\hspace*{-3.2cm}
\subfloat[\label{fig:33sqam}]{\vbox to \ht\tempbox{%
\vfil
\includegraphics[scale=0.38]{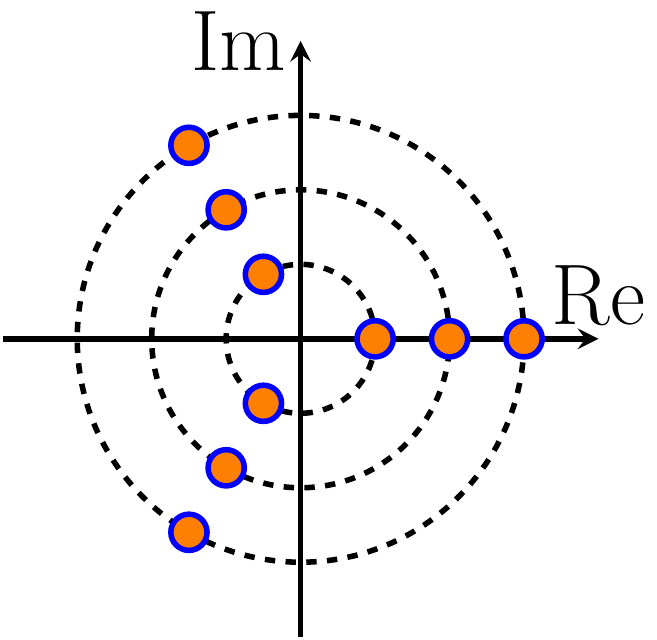}
\vfil}}
\hspace*{-6.3cm}
\subfloat[\label{fig:35sqam}]{\vbox to \ht\tempbox{%
\vfil
\includegraphics[scale=0.38]{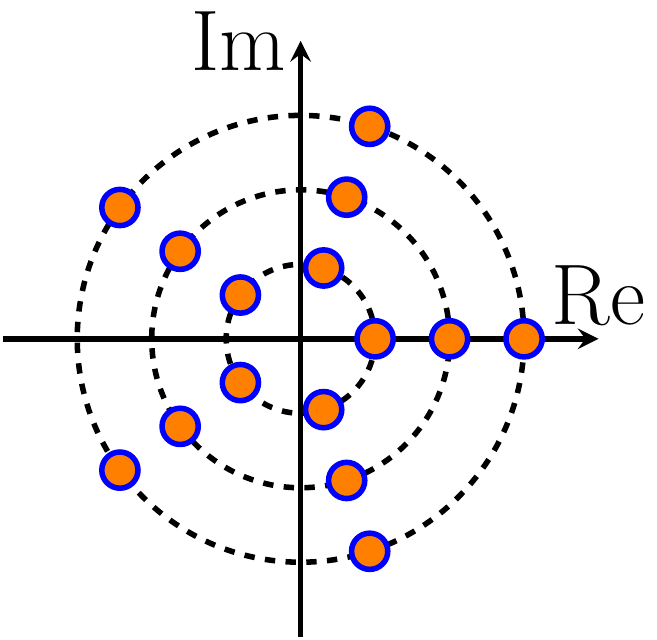}
\vfil}}
\hspace*{-3.2cm}
\subfloat[\label{fig:54sqam}]{\usebox{\tempbox}}\hspace*{-3.2cm}\\
\caption{Three star QAM constellations: (a) $(3,3)$-SQAM, (b) $(3,5)$-SQAM, and (c) $(5,4)$-SQAM.}
\label{fig:sqam}
\end{figure}
\end{subsection}

\begin{subsection}{A Trellis for SQAM Constellations}
\label{subsec:sqam_trellis}

In this subsection, we provide a trellis diagram, $\mathbb{T}$, that
describes the signatures of SLD sequences of length $n$ drawn from
$\krnp$, an $(n_r,n_p)$-SQAM constellation with radius set $\mathcal{D}
= \{ d_1, \ldots, d_{n_r} \}$.  The trellis $\mathbb{T}$ will be used
both to design a codebook $\mathcal{S}$ and for decoding.

The length of $\mathbb{T}$ is $2n-1$.  The vertex set of $\mathbb{T}$ is
given as $
\mathcal{V}=\{v_r\}\cup\mathcal{V}_1\cup\mathcal{V}_2\cup\ldots\cup\mathcal{V}_{2n-2}\cup\{v_g\},
$ where, for $j\in\range{1}{2n-2}$, $\mathcal{V}_j$ is the $n_r$-set
$\mathcal{V}_j=\{v_j^1,v_j^2,\ldots,v_j^{n_r}\}$.  Thus, except at depth
$0$ and at depth $|\mathbb{T}|$, the trellis $\mathbb{T}$ has exactly
$n_r$ vertices at each depth.  To define the edge-label alphabet of
$\mathbb{T}$, for any (not necessarily distinct) $j$ and
$k\in\range{1}{n_r}$ let 
\begin{equation}
  \Psi_{j,k}=\{\psi(a,b) \colon a \in \mathcal{D}_j,~b \in \mathcal{D}_k \}.
\label{eq:Psi_definition}
\end{equation}
By symmetry, $\Psi_{j,k}=\Psi_{k,j}$.  Let
$\mathcal{A}_\text{ISI-present}=\bigcup_{j=1}^{n_r} \bigcup_{k=1}^{n_r}
\Psi_{j,k}$ and let
$\mathcal{A}_\text{ISI-free}=\{d_1^2,\ldots,d_{n_r}^2\}$.  The
edge-label alphabet of $\mathbb{T}$ is then given as
$\mathcal{A}=\mathcal{A}_\text{ISI-present}\cup\mathcal{A}_\text{ISI-free}$.

Finally, to define the edge set of $\mathbb{T}$, let
\begin{align}
\mathcal{E}_r&=\left\{(v_r,d_j^2,v_1^j): j \in \range{1}{n_r}\right\},\nonumber\\ 
\mathcal{E}_g&=\left\{(v_{2n-2}^j,d_j^2,v_g):j \in \range{1}{n_r}\right\},\nonumber\\
\mathcal{E}_\text{ISI-present}&=
\Big\{(v_{2b+1}^j,\zeta,v_{2b+2}^k):b\in\range{0}{n-2},\nonumber\\
&\quad\quad j,k \in \range{1}{n_r}, \zeta\in\Psi_{j,k}\Big\},\nonumber\\
\mathcal{E}_\text{ISI-free}&=
\Big\{(v_{2b}^j,d_j^2,v_{2b+1}^j):b \in \range{1}{n-2},j \in \range{1}{n_r}\Big\}\nonumber;
\end{align}
then
$\mathcal{E}=\mathcal{E}_r\cup\mathcal{E}_g\cup\mathcal{E}_\text{ISI-present}
\cup\mathcal{E}_\text{ISI-free}$.  Note that $\mathcal{E}_r$ is the set
of edges in $\mathbb{T}$ incident from the root vertex $v_r$,
$\mathcal{E}_g$ is the set of edges incident to the goal vertex $v_g$,
$\mathcal{E}_\text{ISI-present}$ are edges incident from vertices at odd
depth, while $\mathcal{E}_\text{ISI-free}$ are edges incident from
vertices at even depth.  There are $|\Psi_{j,k}|$ parallel edges between
$v_{2b+1}^j$ and $v_{2b+2}^k$, for $b\in\range{0}{n-2}$ and $j,k \in
\range{1}{n_r}$.

\begin{example}
\label{example_33}
In this example, we sketch the trellis diagram for the $(3,3)$-SQAM
constellation with radius set
$\mathcal{D}=\{2\sqrt{2},4\sqrt{2},6\sqrt{2}\}$, as shown in
Fig.~\ref{fig:33sqam}, when the length of SLD sequences is $n=3$.  We
have $\Psi_{1,1}=\left\{8,5\right\}$, $\Psi_{1,2}=\left\{19,13\right\}$,
$\Psi_{1,3}=\left\{27,36\right\}$, $\Psi_{2,2}=\left\{20,32\right\}$,
$\Psi_{2,3}=\left\{33,51\right\}$, and
$\Psi_{3,3}=\left\{45,72\right\}$.  The resulting trellis is shown in
Fig.~\ref{fig:trellis_example}.  One may see that for $b\in\{1,3\}$, the
edge labels between, for example, $v_b^{1}$ and $v_{b+1}^{3}$ belong to
$\Psi_{1,3}$. Similarly, the edge label between, for example, $v_2^{2}$
to $v_3^{2}$ is equal to $(4\sqrt{2})^2$.  For ease of reading, the
radius set $\mathcal{D}$ for this example was judiciously chosen to
result in integer edge labels. 
\end{example}

\begin{figure}[t]
\centering
\includegraphics[scale=0.66666]{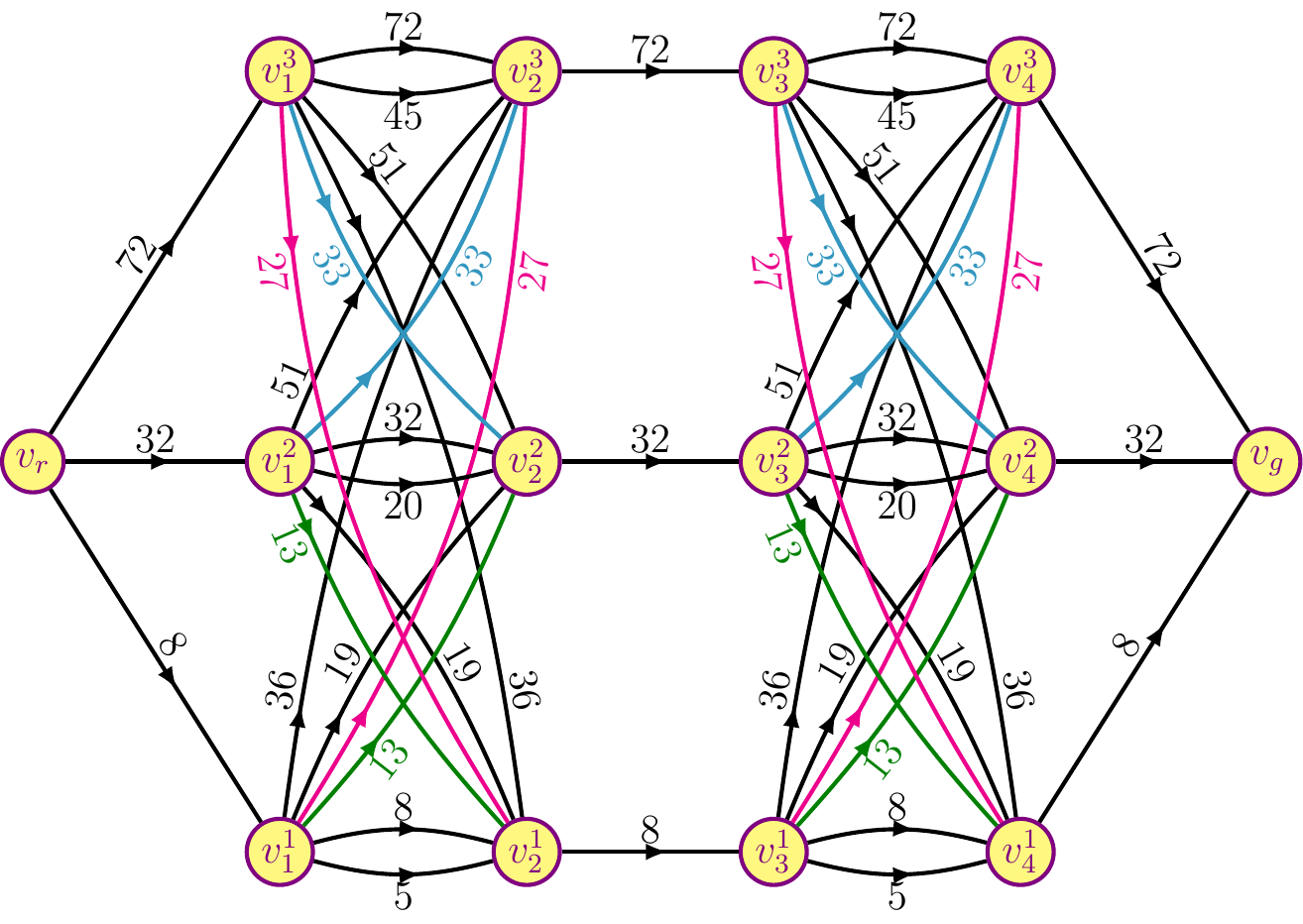}
\caption{The trellis for Example~\ref{example_33}.}
\label{fig:trellis_example}
\end{figure}

Theorem~\ref{thm:psi_cardinality} gives the number of parallel edges
between $v_j^k$ and $v_{j+1}^{\ell}$, for an odd
$j\in\{1,3,\ldots,2n-1\}$.
\begin{theorem}
\label{thm:psi_cardinality}
For any $k$ and $\ell\in\range{1}{n_r}$, $|\Psi_{k,\ell}|=\lceil\frac{n_p+1}{2}\rceil$.
\end{theorem}
\begin{IEEEproof}
The phase difference between any $a\in\mathcal{D}_k$ and
$b\in\mathcal{D}_\ell$ is in $\Phi_{n_p}$ and, since
$\cos(\theta)=\cos(2\pi-\theta)$, for any $\theta \in \mathbb{R}$,
we leave it to the reader
to check that there are exactly
$\lceil\frac{n_p+1}{2}\rceil$ distinct values for $\cos(\arg(ab^\ast))$.
\end{IEEEproof}

For example, Fig.~\ref{fig:coscount} illustrates the situation that
arises when $n_p = 8$ and $n_p = 9$, for which $|\Psi_{k,\ell}| = 5$.

\begin{figure}[t]
\centering
\includegraphics[scale=0.666]{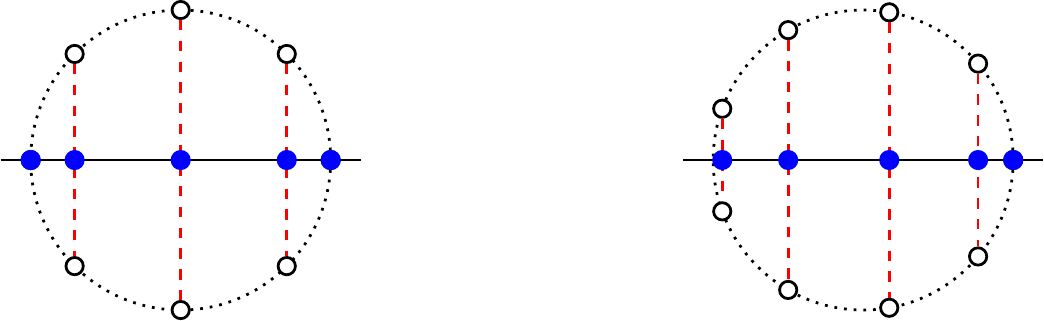}
\caption{Illustrating Theorem~\ref{thm:psi_cardinality}.  The filled
circles show the possible values for $\cos(\arg(ab^\ast))$.}
\label{fig:coscount}
\end{figure}

The following theorem counts $|\mathcal{P}_{\mathbb{T}}|$, the number of
distinct directed paths from root to goal in $\mathbb{T}$.
\begin{theorem}
\label{thm:num_path}
For the trellis $\mathbb{T}$ of length $2n-1$ described in
Sec.~\ref{subsec:sqam_trellis}, 
\begin{equation*}
|\pt|=\left\lceil\frac{n_p+1}{2}\right\rceil^{n-1}n_r^n.
\end{equation*}
\end{theorem}

\begin{IEEEproof}
Let $\bm{A}_j$ denote the adjacency matrix of the $j$th trellis section
of $\mathbb{T}$.  We have
\[
\bm{A}_j = \begin{cases}
\mathds{1}_{1 \times n_r}, & j = 0, \\
\bm{I}_{n_r}, & j \in \{ 2, 4, \ldots, 2n-4 \}, \\
\left\lceil\frac{n_p+1}{2}\right\rceil \mathds{1}_{n_r\times n_r}, & j \in \{ 1, 3,
\ldots, 2n-3 \}, \\
\mathds{1}_{n_r \times 1}, & j = 2n-2, \end{cases}
\]
where the case for odd $j$ follows from
Theorem~\ref{thm:psi_cardinality}.  The number of paths from root to
goal in $\mathbb{T}$ is then given as
\[
|\pt|= \prod_{j=0}^{2(n-1)} \bm{A}_j = 
\left\lceil\frac{n_p+1}{2}\right\rceil^{n-1} n_r^n.
\]
\end{IEEEproof}

The next theorem shows that there is a one-to-one correspondence between
$\mathcal{P}_{\mathbb{T}}$ and $\mathcal{L}_{\mathbb{T}}$.
\begin{theorem}
$|\mathcal{P}_{\mathbb{T}}| = |\mathcal{L}_{\mathbb{T}}|$.
\label{thm:ptlt}
\end{theorem}
\begin{IEEEproof}
The map $\lambda$ taking paths in $\mathcal{P}_{\mathbb{T}}$ to label
sequences in $\mathcal{L}_{\mathbb{T}}$ is surjective by definition.
For $\lambda$ to be injective, every label sequence
$(\ell_0,\ell_1,\ldots,\ell_{2n-2})\in \mathcal{L}_{\mathbb{T}}$ must
have a unique pre-image $\bm{p} \in \mathcal{P}_{\mathbb{T}}$.  The set
of vertices through which $\bm{p}$ traverses is uniquely determined by
$(\ell_0, \ell_2, \ldots, \ell_{2n-2})$.  Furthermore, the selection
from among  parallel edges in the odd trellis sections is uniquely
determined by $(\ell_1, \ell_3, \ldots, \ell_{2n-3})$, since such
parallel edges have distinct labels.
\end{IEEEproof}
\end{subsection}

\begin{subsection}{Finding Square-law Distinct Sequences Using the Trellis}
\label{subsec:sld_trellis}

Finding SLD sequences for a relatively small constellation size and a
small block length $n$ is feasible by computer search.  For example, for
a constellation $\mathcal{K}$, one may produce two matrices,
$\mathcal{Q}\in\mathcal{K}^{|\mathcal{K}|^n\times n}$ and
$\mathcal{M}\in\mathbb{R}^{|\mathcal{K}|^n\times (2n-1)}$ where the rows
of $\mathcal{Q}$ are all possible $n$-vectors defined over $\mathcal{K}$
and the rows of $\mathcal{M}$ are the corresponding signatures.  Then,
one may find all SLD sequences (which are some rows of $\mathcal{Q}$) by
deleting the duplicated rows of $\mathcal{M}$.

For large values of constellation size or block length $n$, finding
those sequences by exhaustive search is computationally intractable.
Fortunately, SLD sequences drawn from an $(n_r,n_p)$-SQAM constellation
can be found using the trellis described in
Sec.~\ref{subsec:sqam_trellis}.

Recall that two $n$-tuples $\bm{s}$ and $\tilde{\bm{s}}$ in $\krnp^n$
are square-law equivalent, written $\bm{s} \equiv \tilde{\bm{s}}$, if
$\Upsilon(\bm{s}) = \Upsilon(\tilde{\bm{s}})$, \textit{i.e.}, if the two
$n$-tuples have the same signature.  The set of such $n$-tuples is
partitioned into disjoint equivalence classes.  A \emph{transversal} of
this collection of classes is then any subset $\mathcal{T} \subset
\krnp^n$ containing exactly one element from each equivalence class.  We
call such a transversal an \emph{SLD-transversal} since, by definition,
the elements of any such transversal are SLD.  A codebook $\mathcal{S}$
is, by definition, a subset of an SLD-transversal; therefore,
SLD-transversals are the largest possible codebooks that can be used for
the system described in Sec.~\ref{sec:system_model}.

\begin{theorem}
\label{thm:unq_label}
For every SLD-transversal $\mathcal{T}$,
\[
\{ \Upsilon(\bm{s}) \colon \bm{s} \in  \mathcal{T} \} =
\mathcal{L}_{\mathbb{T}}.
\]
In other words, the signature of each $n$-vector in an SLD-transversal
is a label sequence in $\mathcal{L}_{\mathbb{T}}$, and each label
sequence in $\mathcal{L}_{\mathbb{T}}$ is the signature of an $n$-vector
in the SLD-transversal.
\end{theorem}

\begin{IEEEproof}
This follows from the definitions of $\mathcal{A}_\text{ISI-free}$ and
$\mathcal{A}_\text{ISI-present}$ and Theorem~\ref{thm:ptlt}.
\end{IEEEproof}

\begin{corollary}
\label{cor:total_number}
The maximum rate that can be achieved by an $(n_r,n_p)$-SQAM
constellation with the system described in Sec.~\ref{sec:system_model}, as $n\rightarrow\infty$,
is $\log_2\left(n_r\left\lceil\frac{n_p+1}{2}\right\rceil\right)$ b/sym. 
\end{corollary}

\begin{proof}
This follows Theorems~\ref{thm:num_path}, \ref{thm:ptlt},
and~\ref{thm:unq_label}.
\end{proof}

Note that the maximum achievable rate for an $(n_r,n_p)$-SQAM
constellation under coherent detection is $\log_2\left(n_rn_p\right)$
and 
\begin{equation*}
0<\log_2\left(n_rn_p\right)-\log_2\left(n_r\left\lceil\frac{n_p+1}{2}\right\rceil\right)<1.
\end{equation*}

According to Theorem~\ref{thm:unq_label}, SLD-transversals are all
mapped to the set $\lt$ by the function $\Upsilon$. Thus every
SLD-transversal is equivalent from the perspective of the detector,
\emph{i.e.}, all of SLD-transversals will have the same error rate.  In
the following, we describe an algorithm that uses the trellis
$\mathbb{T}$ to determine a particular SLD-transversal.

An $n$-vector $(s_0,\ldots,s_{n-1})\in \krnp^n$ is called
\emph{standard} if
\begin{enumerate}
\item $s_0\in\mathbb{R}^{>0}$, and
\item $\arg(s_{\ell+1}s_\ell^\ast)\in[0,\pi], \text{ for every }\ell\in\range{0}{n-2}$.
\end{enumerate}
In other words, the first component of a standard vector is a positive
real number, and the argument of each component is obtained from the
argument of the previous component by adding (modulo $2\pi)$ an angle in
the range $[0,\pi]$.

Define a family of functions
\[
\varphi(\cdot; a,b) \colon  \left[ \frac{3}{8}(a^2 + b^2) - \frac{ab}{4}, \frac{3}{8}(a^2 + b^2) + \frac{ab}{4}\right]
\to [0, \pi],
\]
indexed by positive real numbers $a$ and $b$, and given as
\[
\varphi(\zeta; a,b) = \arccos\left(\frac{8\zeta-3(a^2+b^2)}{2ab}\right).
\]
Given two complex numbers $z_a$ and $z_b$ such that $|z_a| = a$, $|z_b|
= b$, and $\psi(z_a,z_b) = \zeta$, there are at most two possible values
of $\arg( z_a z_b^{\ast})$, namely $\varphi(\zeta; a,b) \in [0, \pi]$
and $2\pi - \varphi(\zeta; a,b) \in [\pi, 2\pi]$.  The function
$\varphi(\cdot; a,b)$ is defined to select the first of these
possibilities.

\begin{theorem}
\label{thm:standard_class_representative}
Each equivalence class (under square-law equivalence)
$\mathcal{C}\subset\krnp^n$ contains exactly one standard vector.
\end{theorem}  

\begin{IEEEproof} 
The signature of each vector in $\mathcal{C}$ is the same, given as,
say, $\bm{\zeta}=(\zeta_0,\ldots,\zeta_{2n-2})\in\mathbb{R}^{2n-1}$.
Let 
\begin{equation} 
\label{eq:s1}
s_0=\sqrt{\zeta_0},
\end{equation} 
and, for $\ell \in \range{0}{n-2}$, let
\begin{equation}
\label{eq:s2} 
s_{\ell+1}=\frac{\sqrt{\zeta_{2\ell}\zeta_{2\ell+2}}}{s^\ast_{\ell}}
\exp\left(i\varphi\left(\zeta_{2\ell+1};\sqrt{\zeta_{2\ell}},\sqrt{\zeta_{2\ell+2}}\right)\right).
\end{equation} 
Let $\bm{s} = (s_0, \ldots, s_{n-1})$.  From (\ref{eq:s1}) and
$(\ref{eq:s2})$ we have $\Upsilon(\bm{s})=\bm{\zeta}$; thus,
$\bm{s}\in\mathcal{C}$.  Furthermore, by construction, $\bm{s}$ is
standard.  Uniqueness follows from the fact that the function $\varphi$
always chooses (from two possibilities) the unique phase difference
between successive components that results in a standard vector.
\end{IEEEproof}

Define $\Upsilon^{-1}: \lt \to \ts$ that, given a signature $(\zeta_0,
\ldots, \zeta_{2n-2}) \in \lt$ returns the standard vector with that
signature, as determined by (\ref{eq:s1}) and (\ref{eq:s2}).

The (unique) SLD-transversal containing only standard vectors, denoted
by $\ts$, is called the \emph{standard SLD-transversal}. The elements of
$\ts$ can be computed by applying $\Upsilon^{-1}$ to the label sequences
obtained by following every possible path from root to goal in the
trellis $\mathbb{T}$.  When $|\ts|$ is sufficiently small, the elements
of $\ts$ be stored in a lookup table, thereby allowing for an efficient
mapping between messages and codewords.   This is the approach followed
for the codebooks $\mathcal{S}$ studied in this paper, which are
obtained as power-of-two subsets of $\ts$.

For larger codebooks, the mapping between messages and codewords becomes
more complicated, possibly requiring enumerative encoding techniques to
select a suitable path through $\mathbb{T}$.  However, once such a path
is selected, computation of a standard vector is readily accomplished by
applying $\Upsilon^{-1}$ to the corresponding label sequence.
\end{subsection}

\begin{subsection}{Trellis Decoding}
\label{subsec:trellis_decoding}

In addition to determining $\ts$, the trellis $\mathbb{T}$ of
Sec.~\ref{subsec:sqam_trellis} can be used with the Viterbi algorithm to
detect the sequence $\hat{\bm{x}}\in\mathcal{S}$ from the noisy received
vectors $\bm{y}=(y_0,\ldots,y_{n-1})$ and $\bm{z}=(z_0,\ldots,z_{n-2})$.
For the reasons of simplicity noted in
Sec.~\ref{subsec:integrate_and_dump}, we use (\ref{eq:y_likelihood}) and
(\ref{eq:z_likelihood}) in branch-metric computation to approximate true
ML detection.

Consider the transmission of the standard vector corresponding to the
label sequence associated with some trellis path
$\bm{p}=(e_0,\ldots,e_{2n-2})\in\pt$.  From (\ref{eq:yk_aprx}),
(\ref{eq:y_bar}), (\ref{eq:zl_aprx}), (\ref{eq:z_bar}), and
Theorem~\ref{thm:unq_label} one may see that, in the absence of noise,
the received values are (to a close approximation) a scaled version of
the edge labels. In particular, in the noise-free case,
	$y_k\simeq\mathfrak{R}\,\beta e^{-\varrho L}\left(E_\text{in}\kappa\alpha_\beta\right)^2\lambda(e_{2k})$
and
	$z_\ell\simeq\mathfrak{R}\,(1-\beta)e^{-\varrho L}\left(E_\text{in}\kappa\alpha_\beta\right)^2
\lambda(e_{2\ell+1})$, for any $k\in\range{0}{n-1}$ and
$\ell\in\range{0}{n-2}$.

Now in the presence of noise, from (\ref{eq:y_likelihood}) and
(\ref{eq:z_likelihood}) it follows that the mean and the variance of the
received $\bm{y}$ and $\bm{z}$ entries are, respectively, linear and
affine functions of the entries of the label sequence $\lambda(\bm{p})$
in our approximation.  In particular, for any edge $e \in \mathcal{E}$,
let $\mathsf{Mean}:\mathcal{E}\rightarrow\mathbb{R}$, and
$\mathsf{Var}:\mathcal{E}\rightarrow\mathbb{R}$ be defined as
\begin{equation*}
	\mathsf{Mean}(e)=\mathfrak{R}\,e^{-\varrho L}\left(E_\text{in}\kappa\alpha_\beta\right)^2\lambda(e),
\end{equation*}
and
\begin{equation*}
	\mathsf{Var}(e)=e^{-\varrho L}\left(E_\text{in}\kappa\alpha_\beta\right)^2\sigma_\text{sh}^2\lambda(e)+T\sigma_\text{th}^2.
\end{equation*}
Thus for an edge $e_{2j}$ in the $(2j$)th trellis section, $j \in
\range{0}{n-1}$, incident from a vertex of even depth, the mean and the
variance of $y_j$ are given as $(1-\beta)\mathsf{Mean}(e_{2j})$ and
$(1-\beta)\mathsf{Var}(e_{2j})$, respectively.  Similarly for an edge
$e_{2j+1}$ in the $(2j+1)$th trellis section, $j\in\range{0}{n-2}$,
incident from a vertex of odd depth, the mean and the variance of $z_j$
are given as $\beta\mathsf{Mean}(e_{2j+1})$ and
$\beta\mathsf{Var}(e_{2j+1})$, respectively.

We define two branch-metric functions,
$w_\text{ISI-free}:\mathbb{R}\times\mathcal{E}\rightarrow\mathbb{R}$,
and
$w_\text{ISI-present}:\mathbb{R}\times\mathcal{E}\rightarrow\mathbb{R}$,
as follows.  Let
\begin{align*}
w_\text{ISI-free}(y,e)=\frac{\big(y-(1-\beta)\mathsf{Mean}(e)\big)^2}{(1-\beta)\mathsf{Var}(e)}+\ln\left(\mathsf{Var}(e)\right),
\end{align*}
and let
\begin{align*}
w_\text{ISI-present}(z,e)=\frac{\big(z-\beta\mathsf{Mean}(e)\big)^2}{\beta\mathsf{Var}(e)}+\ln\left(\mathsf{Var}(e)\right),
\end{align*}
for any $y$ and $z\in\mathbb{R}$ and any $e\in\mathcal{E}$.

To detect $\hat{\bm{x}}$ from $\bm{y}$ and $\bm{z}$, one may run the
Viterbi algorithm on $\mathbb{T}$ with branch metrics defined above to
determine the path $\bm{p}_\text{mw}\in\pt$ of minimum total weight,
computing
\begin{multline*}
\bm{p}_\text{mw}=\underset{(e_0,\ldots,e_{2n-2})\in\pt}{\arg\min}\\
\sum_{j\in\range{0}{n-1}}w_\text{ISI-free}(y_j,e_{2j})+
\sum_{j\in\range{0}{n-2}}w_\text{ISI-present}(z_j,e_{2j+1}).
\end{multline*}
The detected codeword is then
$\hat{\bm{x}}=\Upsilon^{-1}(\lambda(\bm{p}_\text{mw}))$, provided that
$\hat{\bm{x}} \in \mathcal{S}$.  In case $\hat{\bm{x}} \not\in
\mathcal{S}$, a \emph{decoding failure} is declared.  In practice,
$\mathcal{S}$ is selected to have as many codewords from $\ts$ as
possible, while still having a size that is a power of two.  Thus the
majority of paths in $\pt$ are codewords, and decoding failures are
rare.  Indeed, the numerical simulations of Sec.~\ref{sec:simulation} do
indeed confirm that, at the launch powers required to achieve reasonably
small decoding error rates, the probability of decoding failure is
negligible.
\end{subsection}
\end{section}


\begin{section}{Numerical Simulation}
\label{sec:simulation}
In this section, the schemes described in sections
\ref{sec:system_model} and \ref{sec:trellis} are verified by numerical
simulations. 

\begin{subsection}{Simulation Setup}

Because we consider square-law detection with both shot noise and
thermal noise, a notion of signal-to-noise ratio (SNR), as would be
appropriate for linear additive white Gaussian channels, is not
appropriate here.  Accordingly, the various figures of merit considered
here are shown as functions of launch power.  We use practical values
for $\sigma_\text{sh}^2$ and $\sigma_\text{th}^2$ in all cases.

In all simulations we have assumed a transmission length of
$L=10~\text{km}$.  Furthermore, \emph{laser power} refers to the
mean-squared power of the unmodulated waveform, \textit{i.e.},
$\frac{1}{2}E_\text{in}^2$.

In~\cite{tukey} it is shown that $\beta=0.9$ provides near-optimal error
rates and achievable rate performance versus launch power. However,
$\beta=0.9$ requires integration intervals as short as $10\%$ of a
symbol period, which may be difficult to implement at high symbol rates.
To tackle this issue, we have used $\beta=0.5$ in all simulations of
this paper, corresponding to an integration period that is $50\%$ of the
symbol period for both the ISI-free and ISI-present signalling segments.

The practical ease of $\beta=0.5$ comes at the expense of a worse error rate
and a larger bandwidth.  Note that Tukey waveforms are time-limited, thus,
band-unlimited. Therefore, the typical notion of bandwidth for time-limited
waveforms, \textit{e.g.}, sinc and raised-cosine waveforms, does not apply to
Tukey waveforms.  For this reason, we use $95\%$ in-band-energy as the
bandwidth criterion for comparing Tukey waveforms with each other.  While
$w_{0.9}(\cdot)$ has a $15\%$ spectral overhead compared to the minimum
required bandwidth for Nyquist signalling, the spectral overhead of
$w_{0.5}(\cdot)$ is $33.6\%$.  However, spectral efficiency  may not be
an important performance criterion in some applications, \textit{e.g.},
in data-center interconnects, and so
the spectral overhead of Tukey waveforms may not be an important concern
in such applications.
\end{subsection}

\begin{subsection}{The Photodiode}
\label{subsec:photodiodesim}

The simulations in this chapter are done by assuming the use of an
InGaAs p-i-n photodiode at the receiver. This choice gave better
performance than avalanche photodiodes in our simulations.  The reason
for its better performance is that, as explained in
Sec.~\ref{subsec:integrate_and_dump}, the performance of direct
detection schemes depends on the symbol rate and at sufficiently large
symbol rates  shot noise dominates thermal noise. As noted by Agrawal,
``the SNR of APD receivers is worse than that of p-i-n receivers when
shot noise contribution is concerned.''~\cite{agrawal}. 

The two-sided power spectral density of the thermal noise is~\cite{agrawal}
\begin{equation*}
\sigma_\text{th}^2=\frac{2\mathsf{k}T_k}{R_L},
\end{equation*}
where $\mathsf{k}$ is the Boltzmann constant, $T_k$ is the absolute
temperature, and $R_L$ is the load resistor.  In the simulations we have
used $T_k=300$~K and $R_L=300~\Omega$.  Furthermore, the
	two-sided power spectral density of $n_\text{sh}(t)$ of a p-i-n photodiode
is $\sigma_\text{sh}^2=\mathsf{e}\mathfrak{R}$, where $\mathsf{e}$ is
the elementary charge~\cite{agrawal}.  The typical range for the responsivity
of an InGaAs p-i-n photodiode is
$0.6-0.9~\text{mA/mW}$~\cite{agrawal}; in the
simulations, we have used
$\mathfrak{R}=0.75~\text{mA/mW}$.
\end{subsection}

\begin{subsection}{Codebook-size Constraint}
In all of our numerical simulations, we have assumed that $|\mathcal{S}|=M$ is a power of two.  For example,
for the $(8,4)$-SQAM constellation and with a block length of $5$,
we have chosen $2^{21}=2097152$ of the
$\left|\ts\right|=2654208$ available
SLD sequences to form $\mathcal{S}$.
Under this power-of-two restriction, the maximum achievable rate
of the proposed scheme using $(n_r,n_p)$-SQAM with a block length
$n$ is 
\[\frac{1}{n}\left\lfloor\log_2\left(
\left\lceil\frac{n_p+1}{2}\right\rceil^{n-1}n_r^n
\right)\right\rfloor\quad\text{b/sym}.\] 

\end{subsection}

\begin{subsection}{Optimizing Ring Spacing}
\label{subsec:constellation_design}

A pair of $n_r$ and $n_p\in\mathbb{Z}^{\geq 2}$ does not uniquely
specify an $(n_r,n_p)$-SQAM constellation; one must also specify the
radius set $\mathcal{D}$.  In this section, for a fixed $n_r$ and $n_p$
we try to find the best value of $\delta\in\mathbb{R}^{>0}$ such that
\begin{equation}
\mathcal{D}=\left\{1,1+\delta,1+2\delta,\ldots,1+(n_r-1)\delta\right\}.
\label{eq:ring_set}
\end{equation}
Fig.~\ref{fig:8-4_spacing} shows the BER of $(8,4)$-SQAM constellations
for different values of $\delta$. For these constellations
$M=|\mathcal{S}|=2^{12}$ and $n=3$; thus, their maximum achievable rate
is $4$ b/sym. One can see that among the tested $\delta$ values,
$\delta=0.2$ has the best performance.

Table~\ref{tab:ring_spacing} shows the best $\delta$ for different
constellations. The tested values of $\delta$ are similar to those of
Fig.~\ref{fig:8-4_spacing}.  
The criteria for choosing the best $\delta$
is the one which results in the least launch power at BER of $10^{-3}$.
The rows which are not highlighted are those for which we can achieve at
least as great a rate, in b/sym, with another $(n_r,n_p)$-SQAM
constellation at the same symbol rate, yet with smaller launch power.  For
example, the row of $(2,5)$-SQAM constellation is not highlighted as at
$50$~Gbaud it achieves a rate of about $2.2$~b/sym at a launch
power of $-9$~dBm while $(3,2)$-SQAM constellation achieves $2.43$~b/sym
at a launch power of $-12.2$~dBm. Thus, we are interested only in
highlighted rows. 

Theorem~\ref{thm:num_path} implies that the maximum codebook size,
$|\ts|$, depends on $\left\lceil\frac{n_p+1}{2}\right\rceil$. As a
result, for a fixed block length $n$ and a radius set $\mathcal{D}$, an
$(n_r,n_p)$-SQAM constellation with an even $n_p$ results in the same
transversal cardinality as an $(n_r,n_p+1)$-SQAM constellation.
Therefore, as the ring points of the former constellation are further
apart than those of the latter one, we expect the $(n_r,n_p)$-SQAM
constellation to have a better performance than the $(n_r,n_p+1)$
constellation.  Table~\ref{tab:ring_spacing} does indeed support this
claim, as all highlighted rows have an even $n_p$.
\end{subsection}

\begin{figure}[t]
\centering
\includegraphics[scale=0.666666]{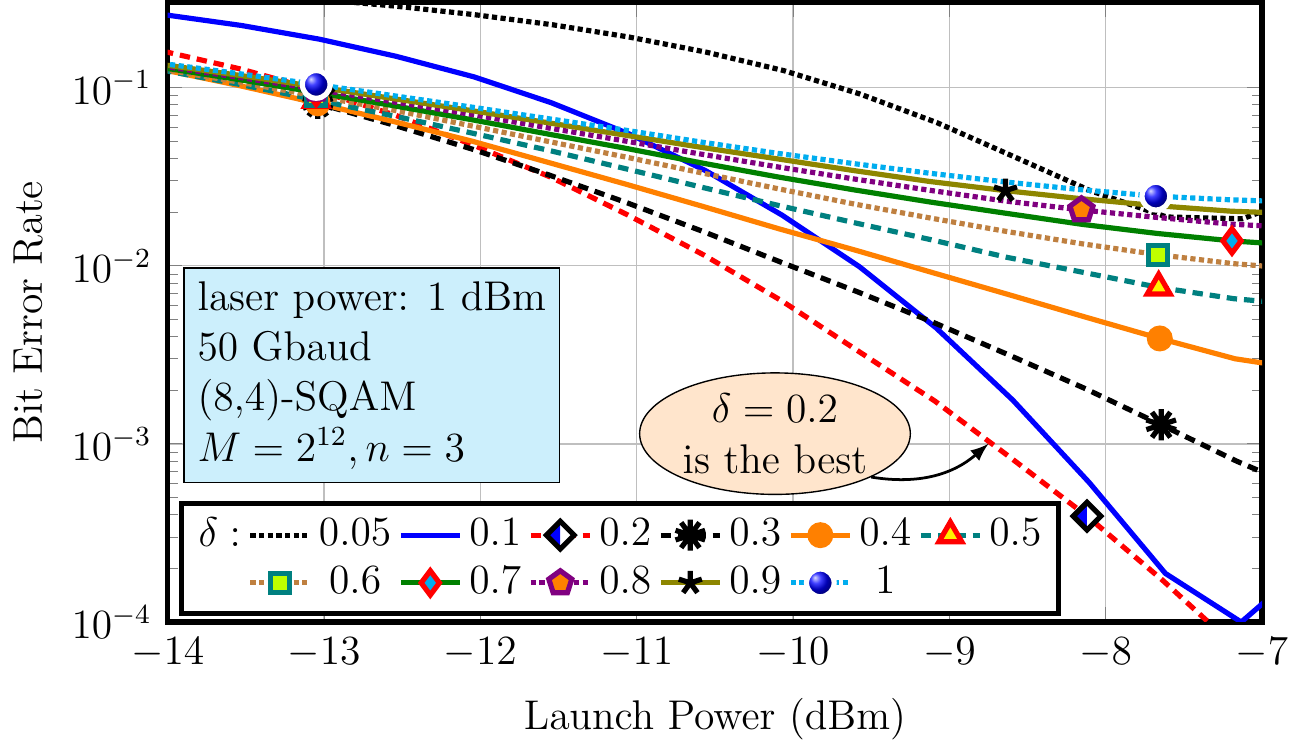}
\caption{BER curves for several $(8,4)$-SQAM constellations with $n=3$.}
\label{fig:8-4_spacing}
\end{figure}

\begin{subsection}{Bit Error Rate}
\label{subsec:ber}

Fig.~\ref{fig:ber_50G} shows the trellis decoding bit error rate of the
proposed scheme for the constellations highlighted in
Table~\ref{tab:ring_spacing} at $50~$Gbaud and with $1~$dBm laser
power.  As it is apparent, the BER curves have two parts: 1) a roll-down
part where the BER decreases with a higher launch power, and 2) a
roll-up part which behaves in the opposite manner.  The reason for the
roll-up part is that the proposed decoder approximates the transmitted
waveform as (\ref{eq:aprx_driver}), while the waveform actually suffers
from nonlinear distortion according to (\ref{eq:driver_signal}).  As
mentioned in~\ref{subsec:integrate_and_dump}, this is a good
approximation when the power of the modulating waveform, $u(t)$, is
relatively small.  For a fixed laser power, the launch power is
controlled by the modulating-waveform power; \textit{i.e.,} a high
launch-power demands a higher modulating-waveform power which degrades
the accuracy of approximating (\ref{eq:driver_signal}) with
(\ref{eq:aprx_driver}).

Fortunately, a practical raw BER, \textit{e.g.,} about $10^{-3}$, can be
achieved in the roll-down area for all considered constellations, except
$(12,4)$-,  $(13,4)$-, and $(16,4)$-SQAM constellations. 
This is the reason that some cells corresponding to those particular constellations
in Table~\ref{tab:ring_spacing} are
left blank. For example, the smallest BER obtained by a $(16,4)$-SQAM constellation
at $50$~Gbaud is about $0.011$, is achieved at a launch power 
about $-7.9~$dBm and with $\delta=0.2$.

As explained in section~\ref{subsec:trellis_decoding}, the power-of-two constraint may
cause the trellis decoder to declare a decoding failure.
Fig.~\ref{fig:decoding_failure} shows the decoding failure probability
for the same constellations as Fig.~\ref{fig:ber_50G}.  The
constellations which are absent in Fig.~\ref{fig:decoding_failure} but
are present in Fig.~\ref{fig:ber_50G} had no decoding failure at the
tested launch power values.  By comparing
Fig.~\ref{fig:decoding_failure} with \ref{fig:ber_50G} one may see that
at launch power values with a practically small BER, \textit{e.g.},
$10^{-3}$, the decoding failure rate is negligible.

As noted in Sec.~\ref{subsec:integrate_and_dump}, unlike a typical AWGN
communication scheme, the performance under direct detection depends on
the symbol rate. This fact is illustrated in Fig.~\ref{fig:ber_25G} which
provides the BER curves in the same setup as Fig.~\ref{fig:ber_50G},
except that it is at $25$~Gbaud.  By comparing these two figures
one may see that for a fixed constellation, a BER of $10^{-3}$ can be
achieved with approximately $1.5~$dB less launch power at $25~$Gbaud
than at $50~$Gbaud.

Note that in computing BER we treat decoding failures as equivalent to a
momentary BER of $0.5$.  We made no attempt to optimize (\textit{e.g.},
by Gray labelling) the mapping of sequences $\bm{x}\in\mathcal{S}$ to
binary vectors of length $\log_2(|\mathcal{S}|)$.
\end{subsection}

\begin{table}[t]
\centering
\caption{Optimal ring spacing for different star-QAM constellations.}
\label{tab:ring_spacing}
\setlength\tabcolsep{1.1pt}
\begin{tabular}{|c|c|c|c|c|c|}
\hline
\rotatebox{90}{$\bm{(n_r,n_p)}$}& 
$\bm{n}$&
\rotatebox{90}{\textbf{max rate~(b/sym)}} & 
\rotatebox{90}{\begin{tabular}{c}\textbf{launch power}\\\textbf{(dBm)}\\\textbf{at 50~Gbaud}\end{tabular}}& 
\rotatebox{90}{\begin{tabular}{c}\textbf{launch power}\\\textbf{(dBm)}\\\textbf{at 25~Gbaud}\end{tabular}}& 
\rotatebox{90}{\textbf{optimal $\bm{\delta}$}}\\
\hline
\hline
$(2,3)$	&$5$	&$1.8$	&$-12.6$ 	&$14.1$	&$0.2$ \\
\hline
\rowcolor{pink}	$(2,2)$	&$7$	&$1.86$ 	& $-13.25$	&$-14.75$	&$0.2$\\
\hline
$(2,4)$	&$4$	&$2$ 	& $-10.6$	&$-12$	&$0.2$\\
\hline
$(2,5)$	&$5$	&$2.2$	& $-9$	&$-10.5$	& $0.2$\\
\hline
$(3,3)$	&$4$	&$2.25$	&$-12.1$	&$-13.5$	&$0.2$\\
\hline
\rowcolor{pink}	$(3,2)$ &$7$	&$2.43$ & $-12.2$ &$-13.6$ & $0.2$\\ 
\hline 
$(2,6)$	&$4$	&$2.5$	& $-8.2$	&$-9.7$	&$0.1$\\
\hline
$(4,3)$	&$4$	&$2.75$	&$-11.8$	&$-13.3$	&$0.2$\\
\hline
$(3,5)$	&$4$	&$2.75$	& $-9$	&$-10.4$	&$0.1$\\
\hline
$(3,4)$	&$4$	&$2.75$	& $-10.3$	&$-11.8$	&$0.2$\\
\hline
\rowcolor{pink}	$(4,2)$ &$6$	&$2.83$ & $-12.2$ &$-13.6$ 	&$0.2$\\
\hline
\rowcolor{pink}	$(5,2)$ &$5$	&$3$	&$-11.2$&$-12.7$	&$0.2$\\
\hline
$(5,3)$	&$4$	&$3$	&$-11$	&$-12.5$	&$0.2$\\
\hline
$(6,3)$	&$3$	&$3$	&$-10.9$	&$-12.3$	&$0.2$\\
\hline
$(4,4)$	&$4$	&$3$	&$-10$	&$-11.6$	&$0.2$\\
\hline
$(4,5)$	&$3$	&$3$	&$-8.8$	&$-10.3$	&$0.1$\\
\hline
$(3,6)$	&$4$	& $3$	& $-8.2$	&$-9.7$	&$0.1$\\
\hline
\rowcolor{pink}	$(6,2)$ &$5$	& $3.2$ & $-10.8$	&$-12.3$	&$0.2$\\
\hline
\rowcolor{pink}	$(7,2)$	&$3$	&$3.33$	&$-10.3$	&$-11.7$	&$0.2$\\
\hline
$(7,3)$	&$3$	&$3.33$	&$-10.2$	&$-11.6$	&$0.2$\\
\hline
$(5,4)$	&$3$	&$3.33$	&$-9.8$	&$-11.3$	&$0.2$\\
\hline
$(6,4)$	&$3$	&$3.33$	&$-9.7$	&$-11.1$	&$0.2$\\
\hline
$(5,5)$	&$3$	&$3.33$	&$-8.6$	&$-9.9$	&$0.1$\\
\hline
$(6,5)$	&$3$	&$3.33$	&$-8.4$	&$-9.8$	&$0.1$\\
\hline
\end{tabular}
\begin{tabular}{|c|c|c|c|c|c|}
\hline
\rotatebox{90}{$\bm{(n_r,n_p)}$}& 
$\bm{n}$&
\rotatebox{90}{\textbf{max rate~(b/sym)}} & 
\rotatebox{90}{\begin{tabular}{c}\textbf{launch power}\\\textbf{(dBm)}\\\textbf{at 50~Gbaud}\end{tabular}}& 
\rotatebox{90}{\begin{tabular}{c}\textbf{launch power}\\\textbf{(dBm)}\\\textbf{at 25~Gbaud}\end{tabular}}& 
\rotatebox{90}{\textbf{optimal $\bm{\delta}$}}\\
\hline
\hline
$(4,6)$	&$3$	&$3.33$	&$-8$	&$-9.7$	&$0.1$\\
\hline
$(5,6)$	&$3$	&$3.33$	&$-8$	&$-9.6$	&$0.1$\\
\hline
$(7,2)$	&$5$	&$3.6$	&$-10$	&$-11.5$	&$0.2$\\
\hline
$(9,3)$	&$3$	&$3.67$	&$-9.3$	&$-10.8$	&$0.2$\\
\hline
$(8,3)$	&$3$	&$3.67$	&$-9.9$	&$-11.4$	&$0.2$\\
\hline
$(10,3)$	&$3$	&$3.67$	&$-9.1$	&$-10.6$	&$0.2$\\
\hline
$(7,5)$	&$3$	&$3.67$	&$-8.1$	&$-9.4$	&$0.1$\\
\hline
$(7,4)$	&$3$	&$3.67$	&$-9.3$	&$-10.8$	&$0.2$\\
\hline
\rowcolor{pink}	$(8,2)$	&$5$	&$3.8$	&$-10.2$&$-11.7$	&$0.2$\\ 
\hline
$(8,3)$	&$5$	&$3.8$	&$-9.9$	&$-11.4$	&$0.2$\\
\hline
\rowcolor{pink}	$(8,4)$	&$3$	&$4$	&$-8.8$	&$-10.3$	&$0.2$\\
\hline
$(11,3)$	&$3$	&$4$	&$-8.5$	&$-9.9$	&$0.2$\\
\hline
$(9,4)$	&$3$	&$4$	&$-8.6$	&$-10.1$	&$0.2$\\
\hline
$(12,3)$	&$3$	&$4$	&$-8.5$	&$-9.9$	&$0.2$\\
\hline
$(9,5)$	&$3$	&$4$	&$-7.6$	&$-8.9$	&$0.1$\\
\hline
$(8,5)$	&$3$	&$4$	&$-7.9$	&$-9.3$	&$0.1$\\
\hline
$(7,6)$	&$3$	&$4$	&$-7.7$	&$-9.2$	&$0.1$\\
\hline
\rowcolor{pink}	$(8,4)$	&$5$	&$4.2$	&$-8.4$	&$-9.9$	&$0.2$\\
\hline
\rowcolor{pink}	$(10,4)$	&$3$	&$4.33$	&$-8$	&$-9.5$	&$0.2$\\
\hline
$(11,4)$	&$3$	&$4.33$	&$-8$	&$-9.5$	&$0.2$\\
\hline
$(12,4)$	&$3$	&$4.33$	&---	&$-9.2$	&$0.2$\\
\hline
$(8,6)$	&$3$	&$4.33$	&$-7.4$	&$-9$	&$0.1$\\
\hline
\rowcolor{pink}	$(13,4)$	&$3$	&$4.67$	&---	&$-8.6$	&$0.2$\\
\hline
$(16,4)$	&	$3$	&$5$	&---	&---	&$0.2$\\
\hline
\end{tabular}
\end{table}

\begin{figure}[t]
\centering
\includegraphics[scale=0.6666666]{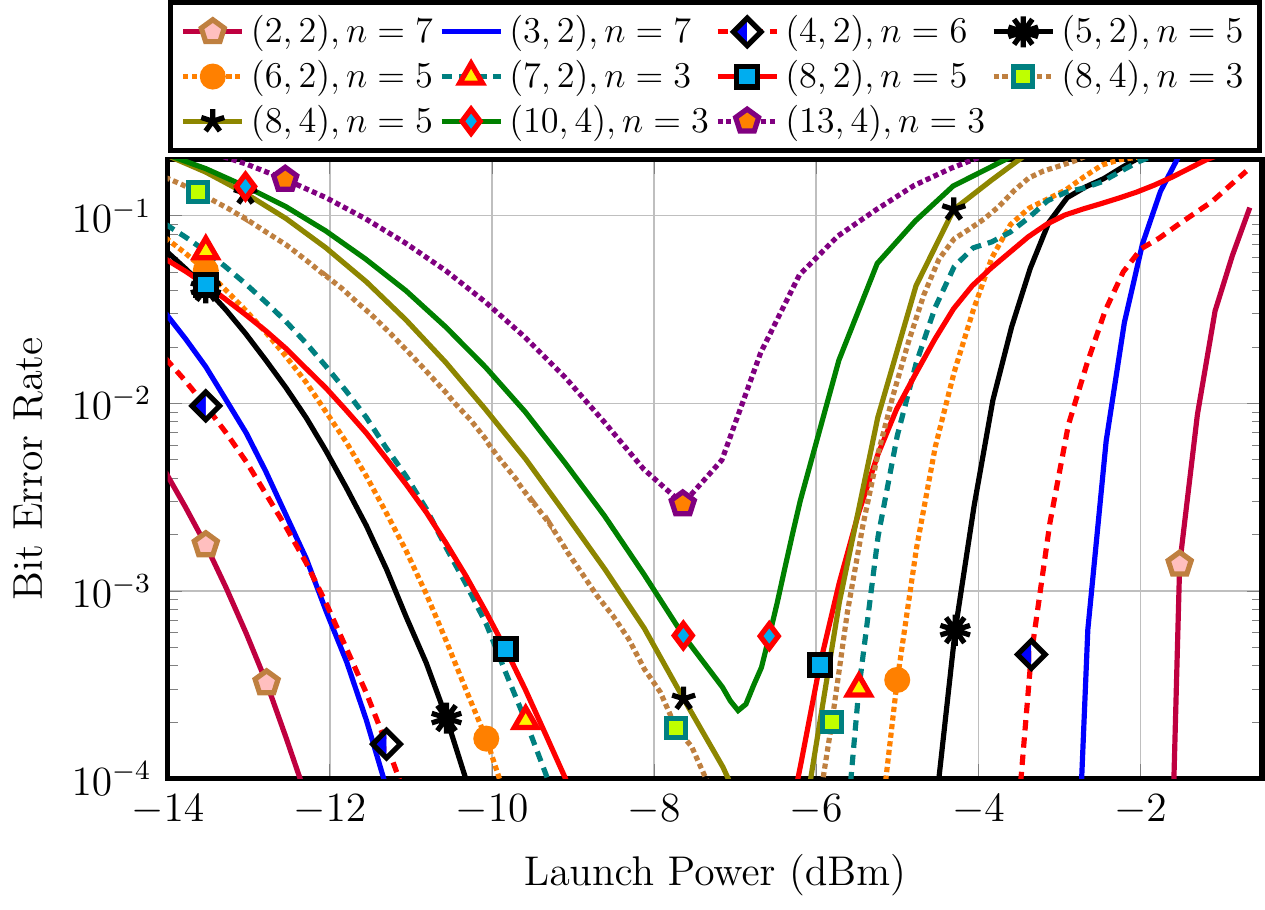}
\caption{BER curves for the constellations highlighted in 
Table~\ref{tab:ring_spacing} at $50$~Gbaud and with $1~$dBm laser power.}
\label{fig:ber_50G}
\end{figure}
\begin{figure}
\centering
\includegraphics[scale=0.666666]{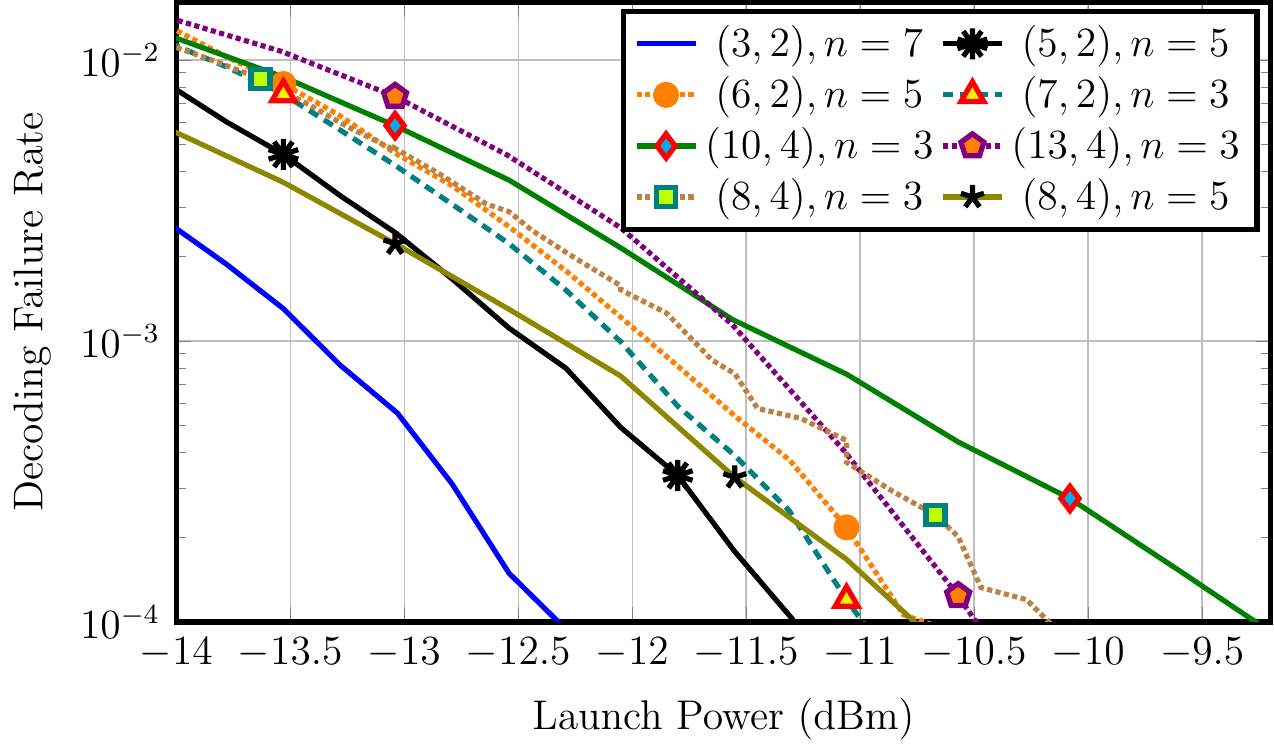}
\caption{Decoding failure rate for the constellations highlighted in
Table~\ref{tab:ring_spacing}, at $50$~Gbaud and with $1~$dBm
laser power. The constellations that are not present in this figure
had no decoding failure at the tried launch powers.}
\label{fig:decoding_failure}
\end{figure}

\begin{subsection}{Achievable Data Rate}
\label{subsec:achievable_rate}

Fig.~\ref{fig:achievable_rate} shows the theoretically maximum achievable data
rate for a few SQAM constellations at $50$~Gbaud and with a $1$~dBm laser
power, obtained by Monte Carlo simulations.  Except the $(16,4)$-SQAM
constellation with block length $n=3$, the remaining constellations are chosen
from the ones highlighted in Table~\ref{tab:ring_spacing} but possibly with a
different block length $n$.  For the $(16,4)$-SQAM constellation we have
$\delta=0.2$ in (\ref{eq:ring_set}).

Similar to the BER figures, the rate figures are obtained by the
power-of-two constraint on $|\mathcal{S}|$.  Without this constraint,
the maximum achievable rate of each constellation can be increased up to
$\frac{1}{n}$ bits.  For example, for the $(5,2)$-SQAM constellation
with $n=4$, the maximum achievable rate by using all SLD sequences is
$\log_2|\ts|\simeq 3.07$~b/sym.  However, as only $2^{12}=4096$ symbols
have been chosen out of $5000$ SLD sequences, the rate has decreased to
$3$~b/sym.

Fig.~\ref{fig:achievable_rate} shows the necessity of using error
correcting codes from two perspectives. First, a desired rate can be
achieved using a higher-order constellation along with an error
correcting code at a lower launch power. For example, while the uncoded
$(8,4)$-SQAM constellation with $n=3$ can achieve $4$~b/sym at a launch
power of $-9.5$~dBm, the same rate can be achieved with a $(10,4)$-SQAM
constellation with $n=3$ and an error correcting code of rate
$\frac{12}{13}$ at a launch power of $-12$~dBm; resulting in a coding
gain of about $2.5$~dBm. Secondly, at a fixed launch power, a higher
data rate can be achieved by using a higher-order constellation along
with an error correcting code. For example, at a launch power of
$-12$~dBm, the uncoded $(5,2)$-SQAM constellation with $n=4$ can achieve
a rate of $3$~b/sym, while by using the $(8,4)$ constellation with $n=3$
along with an error correcting code of rate $\frac{15}{16}$ one may
achieve $3.75$~b/sym. 

\begin{figure}[t]
\centering
\includegraphics[scale=0.6666666]{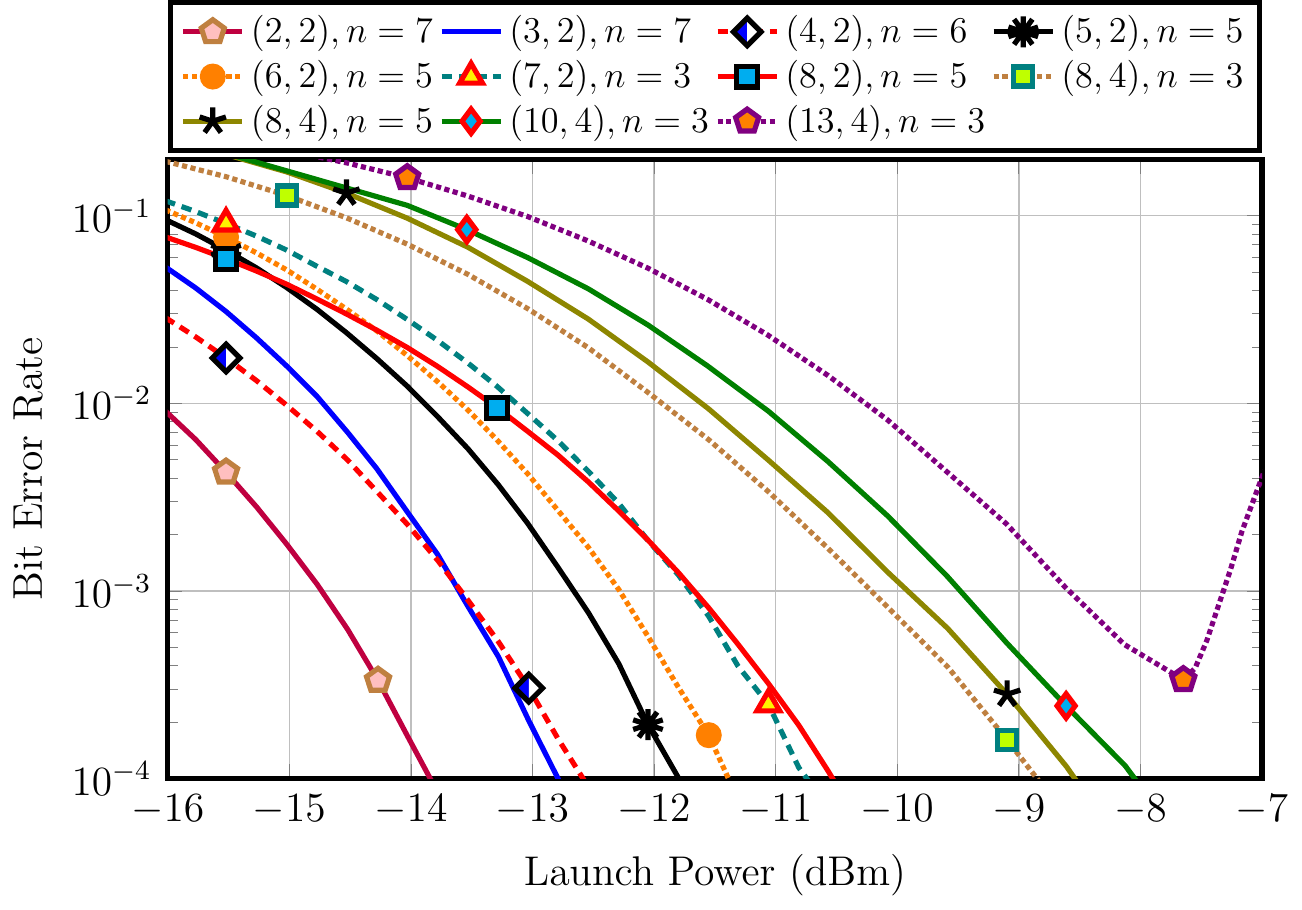}
\caption{BER curves for the constellations highlighted in 
Table~\ref{tab:ring_spacing} at $25$~Gbaud and with $1~$dBm laser power.}
\label{fig:ber_25G}
\end{figure}
\begin{figure}
\centering
\includegraphics[scale=0.6666666]{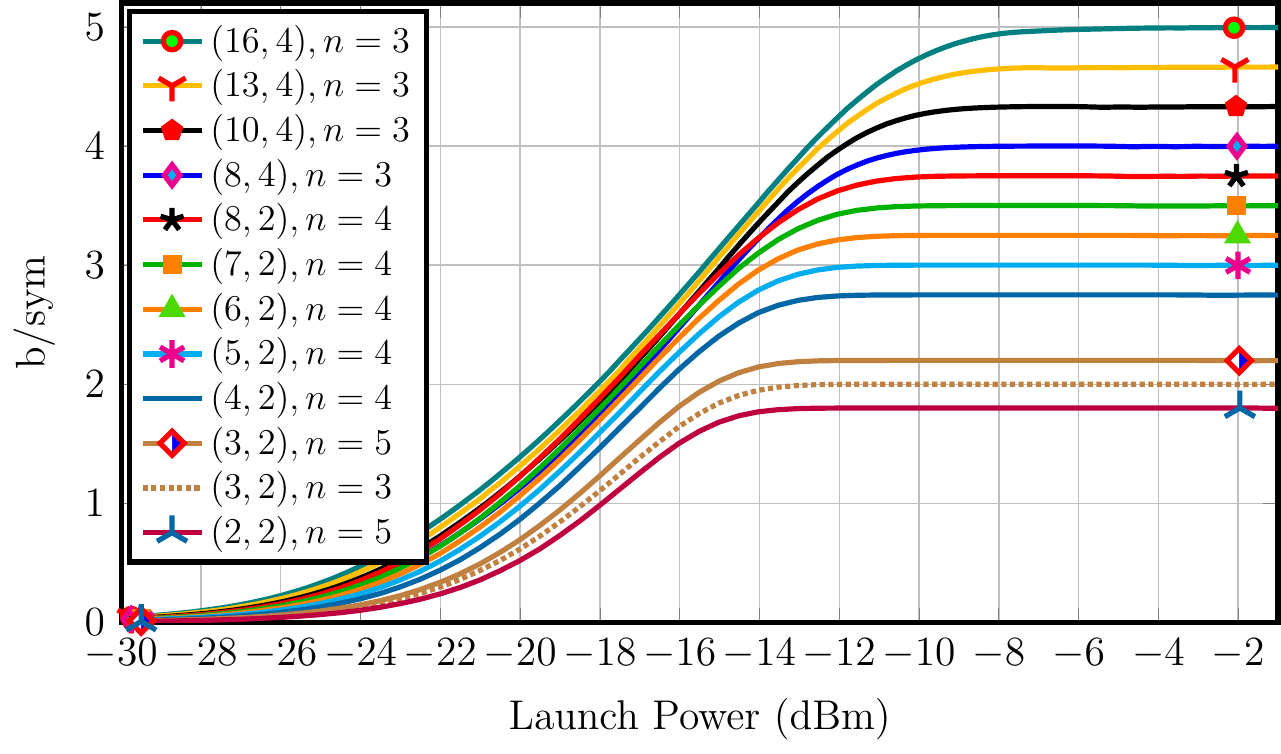}
\caption{Achievable data rate for different star-QAM constellations. 
Except $(16,4)$-SQAM constellation, other constellations are highlighted
in Table~\ref{tab:ring_spacing}, but with a possibly different block length. 
Curves are obtained at $50$~Gbaud and with $1$~dBm laser power.}
\label{fig:achievable_rate}
\end{figure}

Another fact which is shown in Fig.~\ref{fig:achievable_rate} is the
role of the block length $n$ on the performance. While the maximum
achievable rate for the $(3,2)$-SQAM constellation is $2$~b/sym with a
block length of $n=3$, it is $2.2$~b/sym with $n=5$. Note that, in
contrast to~\cite{tukey}, where brute-force search was used, the
increase in data rate resulting from  a larger $n$ does not come with a
huge increase in decoding complexity as the complexity of trellis
decoding is linear in $n$.

\end{subsection}

\begin{figure*}[t]
\centering
\subfloat[\label{fig:LO_ber_84}]{
\includegraphics[scale=0.6666666]{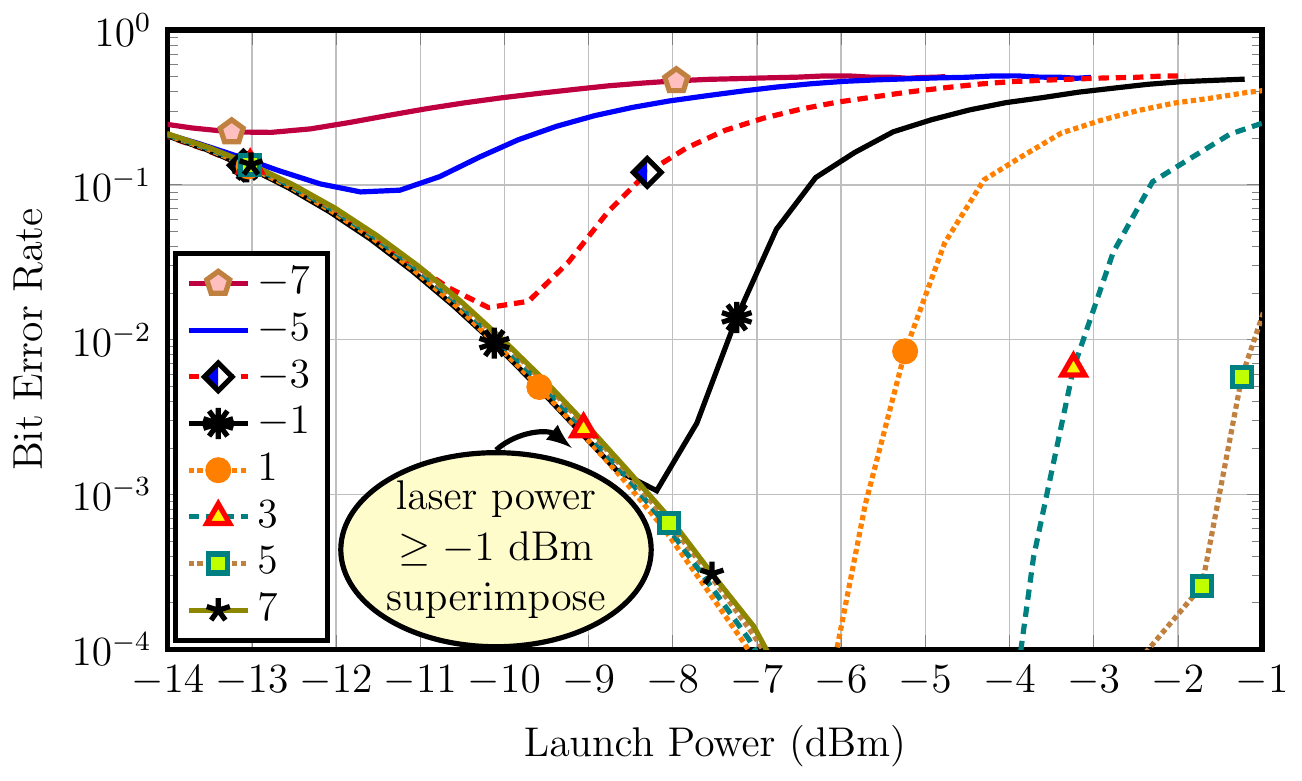}
}
\subfloat[\label{fig:LO_fail_84}]{
\includegraphics[scale=0.6666666]{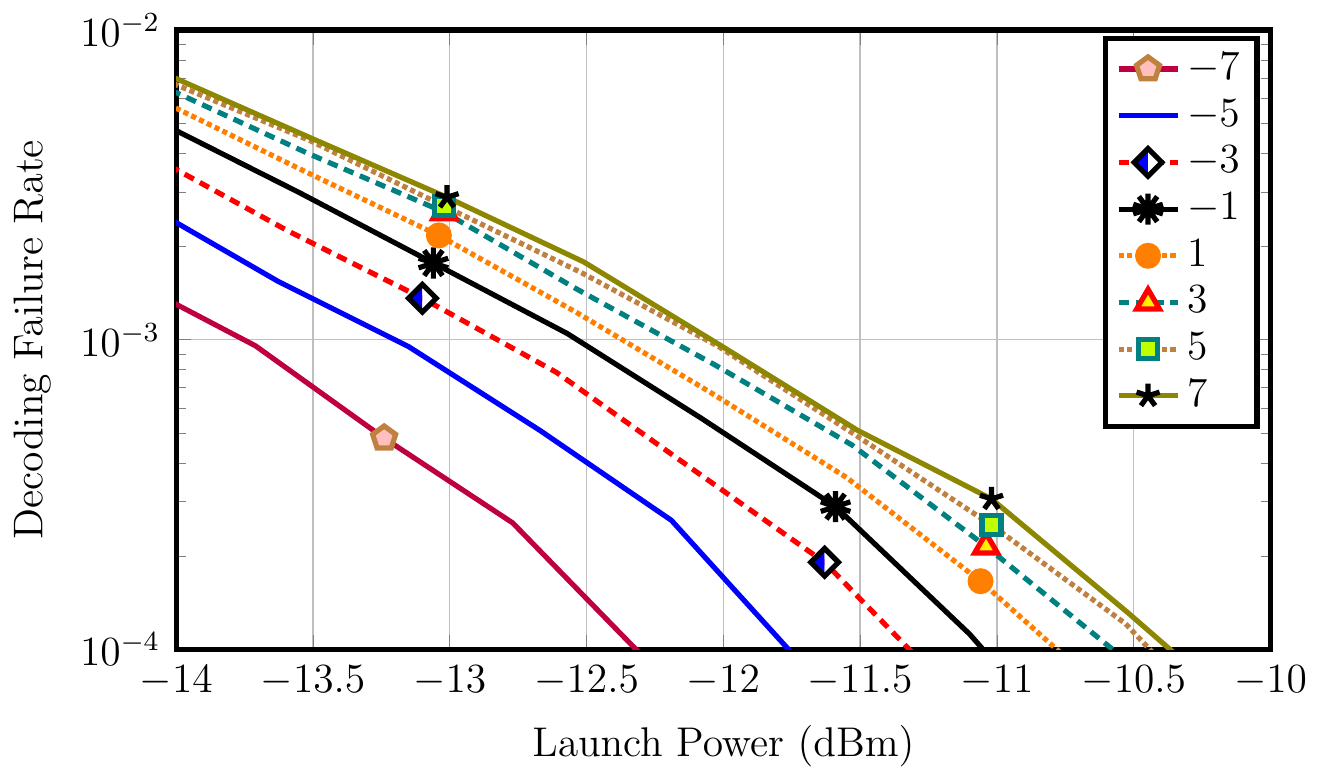}
}
\caption{The BER and decoding failure rate for the $(8,4)$-SQAM constellation with 
$n=5$, $\delta=0.2$ (see (\ref{eq:ring_set})), and at $50~$Gbaud,
at different laser powers. Legends
show the laser power in dBm.}
\label{fig:LO_84}
\end{figure*}

\begin{subsection}{Effect of Laser Power}
\label{subsec:laser_power}

Fig.~\ref{fig:LO_ber_84} shows the BER of the $(8,4)$-SQAM constellation
with $\delta=0.2$, as recommended by Table~\ref{tab:ring_spacing},
$n=5$, and at $50~$Gbaud. Note that the BER curves corresponding to
a laser power of $\geq -1~$dBm superimpose in the region of interest,
\textit{i.e.}, the region with BER about $10^{-3}$. However, the BER
performance degrades as laser power drops further. The reason for this
behavior is again in approximating (\ref{eq:driver_signal}) with
(\ref{eq:aprx_driver}). For a fixed launch power, a lower laser power
demands a higher modulating-waveform power and it violates the condition
for an accurate approximation of (\ref{eq:driver_signal}) by
(\ref{eq:aprx_driver}). 

From Fig.~\ref{fig:LO_ber_84} one may conclude that the laser power
determines the BER at which the roll-down part of the BER curve meets
the roll-up part.  In other words, the larger is the laser power, the
smaller is the BER at the transition point. For example, while the
roll-down to roll-up transition occurs at a BER of about $2\times
10^{-2}$ for a laser power of $-3~$dBm, it happens at a BER of $10^{-3}$
and $<10^{-4}$ for a laser power of $-1~$dBm and $1~$dBm, respectively.
Therefore, depending on the minimum required uncoded BER, one may
determine the minimum required laser power.

Fig.~\ref{fig:LO_fail_84} shows the decoding failure rate with the same
simulation setup as in Fig.~\ref{fig:LO_ber_84}. Interestingly, in
contrast to the BER curves, the smaller is the laser power, the smaller
is the decoding failure rate. However, the smaller decoding failure rate
comes with a higher bit error rate. From the spacing between the
decoding failure rate curves it is apparent that the curves approach a
limiting curve as the laser power increases.  However, note that launch
power values which result a target BER, \textit{e.g.}, about $10^{-3}$,
the decoding failure rate is almost negligible.

Fig.~\ref{fig:LO_ber_22} shows the BER curves for the $(2,2)$-SQAM
constellation with $\delta=0.2$ and block length $n=7$.  The maximum
achievable rate for this constellation and $n$ is $\frac{13}{7}$ b/sym.
By comparing Fig.~\ref{fig:LO_ber_84} with Fig.~\ref{fig:LO_ber_22} one
may see that increasing constellation size, thus rate, increases the
minimum required laser power to get a target BER.  For example, while we
can achieve a BER of $10^{-3}$ with a $-5~$dBm or even $-7$~dBm laser
power with the described $(2,2)$-SQAM constellation, the minimum
required laser power to achieve the same BER with the described
$(8,4)$-SQAM constellation is $-1~$dBm. 

In short-haul applications, \textit{e.g.}, intra data-center
communication, power consumed by lasers play are an important component
of the overall power budget.  One may decide to allocate less power to
lasers in the price of a lower data rate, in b/sym, or higher BER for a
fixed constellation.

\begin{figure}
\includegraphics[scale=0.66666666]{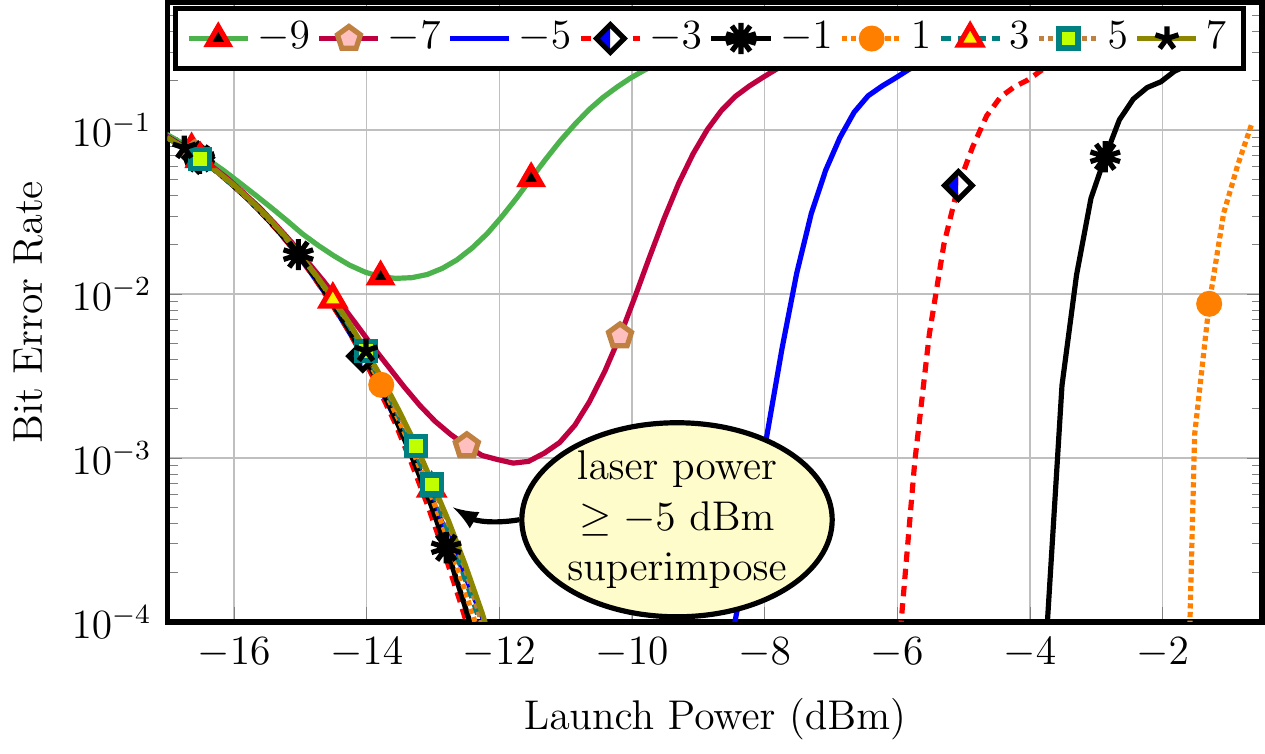}
\caption{The BER for the $(2,2)$-SQAM constellation with $n=7$ and 
$\delta=0.2$ (see (\ref{eq:ring_set})), and at $50~$Gbaud, at 
different laser powers. Legends show the laser power in dBm.}
\label{fig:LO_ber_22}
\end{figure}
\end{subsection}

\end{section}

\begin{section}{Comparing with IMDD}
\label{sec:compare}

In contrast to the majority of the literature, instead of experimental
results all of our figures of merit are obtained by computer-based
simulations.  There are many practical issues, \textit{e.g.}, splicing,
scattering, etc.,  which have been ignored in our numerical simulations.
Therefore, to have a fair comparison, we have simulated IMDD using the
same photodiode parameters that were used in our system. 

Fig.~\ref{fig:imdd_compare} shows the throughput of the proposed and the IMDD
schemes at $50~$Gbaud for a single wavelength and a single polarization.
While the $(8,4)$-SQAM Tukey signalling achieves a throughput of $200~$Gb/s at
a launch power of about $-10~$dBm, IMDD with PAM-$16$ achieves this rate at
about $0~$dBm; thus, by using the proposed scheme a again of about $8~$dB can
be achieved compared to the IMDD scheme.  We can achieve the same throughput by
using a $(16,4)$-SQAM constellation with $n=3$ and an error correcting code of
rate $0.8$ at a launch power about $-12.9~$dBm, as well, \textit{i.e.}, a
$2.9$~dB coding gain.  Furthermore, Fig.~\ref{fig:imdd_compare} shows that for
a fixed launch power and a fixed symbol rate one may achieve higher throughputs
by using the proposed scheme rather than IMDD.  For example, while a throughput
of $200$~Gb/s is achievable with the proposed scheme at a launch power of
$-10~$dBm and at $50~$Gbaud with an $(8,4)$-SQAM constellation of block length
$n=3$, one may achieve only about $145~$Gb/s with an IMDD scheme using PAM-$8$
constellation at the same symbol rate and the same launch power. Note that these
two constellations have the same number of magnitude levels; the proposed
scheme has a higher throughput since it is able to extract
phase information from
a complex-valued constellation.


\begin{figure}[t]
\centering
\includegraphics[scale=0.6666666]{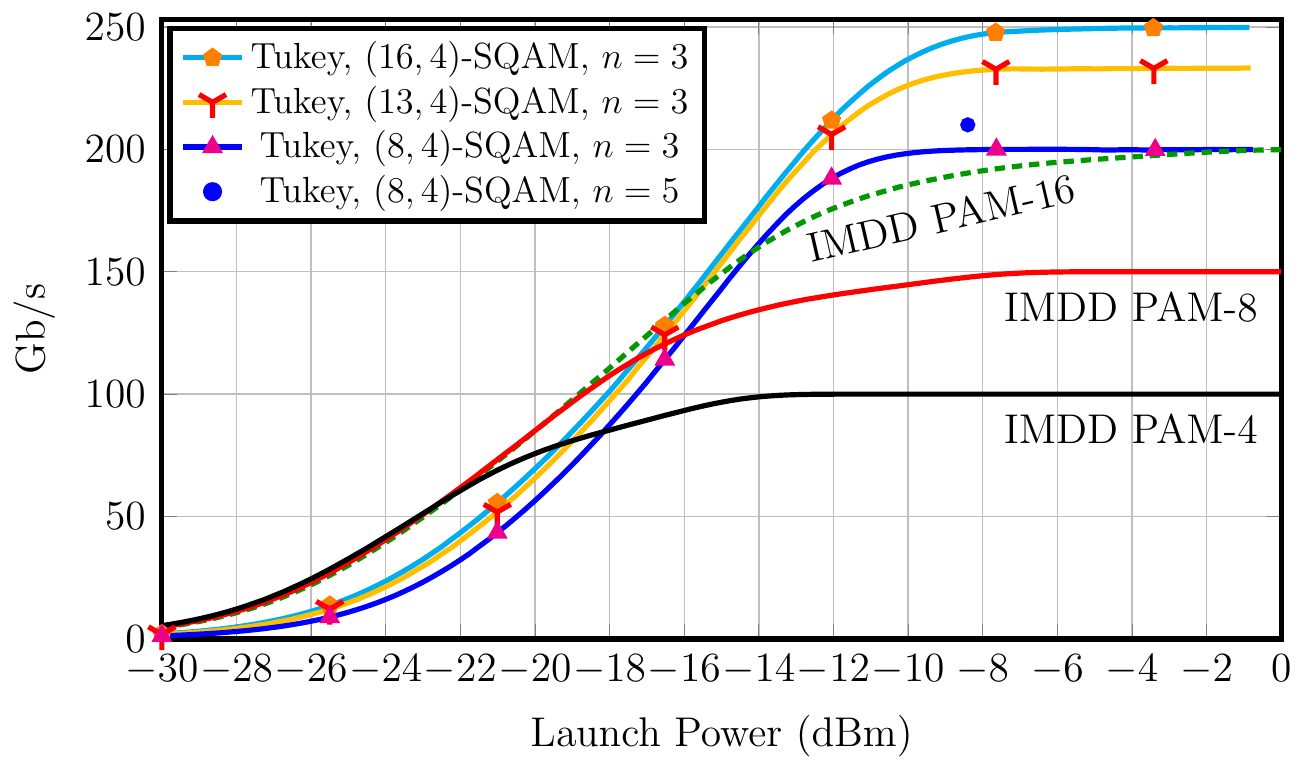}
\caption{Throughput for IMDD and the proposed scheme at $50~$Gbaud.}
\label{fig:imdd_compare}
\end{figure}

\end{section}


\begin{section}{O-band Operation}
\label{sec:Oband}

So far we have assumed operation in the C band according to the
system model shown in Fig.~\ref{fig:system_model}.
Because of the non-negligible chromatic dispersion
in this band and the difficulty of post-detection dispersion compensation
after signal detection using
a single photodiode, dispersion was (partially) precompensated at the
transmitter, as discussed in Sec.~\ref{subsec:dispersion_precompensation}.
However, pre-compensation of chromatic dispersion requires
knowledge of the link length, which is generally not possible without
a feedback channel.

To remedy this issue, one may operate in the O band,
near the zero-dispersion wavelength of
SSMF, omitting transmitter-side dispersion precompensation, and
making the transmitter agnostic to the transmission distance.
This changes the system model to that of
Fig.~\ref{fig:system_model_Oband}. The advantage of having a link-length
agnostic transmitter comes at the expense of having channel with
greater loss compared to a C-band channel.  In particular, the fiber loss
in O band is about $0.5~\text{dB}\cdot\text{km}^{-1}$. 

Fig.~\ref{fig:Oband-ber} shows the O-band BER for the constellations
highlighted in Table~\ref{tab:ring_spacing}, using a laser power of
$2~\text{dBm}$ and at a symbol rate of $50~\text{Gbaud}$. 
We do not assume operation with \emph{exactly} zero dispersion;
instead we simulate operation at two wavelengths 
assumed to have chromatic dispersions of $-1$ and
$1~\text{ps}\cdot\text{nm}^{-1}\cdot\text{km}^{-1}$, respectively,
leaving the channel with some uncompensated residual dispersion.

For a $10~\text{km}$ length of SSMF, the total power loss of the transmission
link in the O band is $3~\text{dB}$ more than the total power loss in the C
band.  We see this by comparing the roll-down parts of the curves in
Fig.~\ref{fig:Oband-ber} with their corresponding curves in
Fig.~\ref{fig:ber_50G}. For example, one may see from Fig.~\ref{fig:Oband-ber}
that the $(8,4)$-SQAM constellation with block length $n=3$ achieves a BER of
$10^{-3}$ at a launch power of $-5.8~\text{dBm}$ in the O band; while, from
Fig.~\ref{fig:ber_50G}, this constellation achieves the same BER at a launch
power of $-8.8~\text{dB}$ in the C band. 

Note that this ``$3~\text{dB}$ difference'' is violated for
the $(10,4)$- and $(13,4)$-SQAM constellations.
These large constellations achieve a target BER at a
higher launch power compared to the other smaller constellations highlighted in
Table~\ref{tab:ring_spacing}.  Furthermore, due to the greater loss
in the O band, a higher launch power is required to achieve a
target BER. However, as explained in Sec.~\ref{subsec:integrate_and_dump} and
for a fixed laser power, the higher launch power results in a higher modulator
nonlinearity. Thus, it is the nonlinearity of the modulator that
hinders the performance
of these large constellations in the O band.  This issue
may be remedied by operating at a higher laser power.

\begin{figure}
\includegraphics[scale=0.66666666]{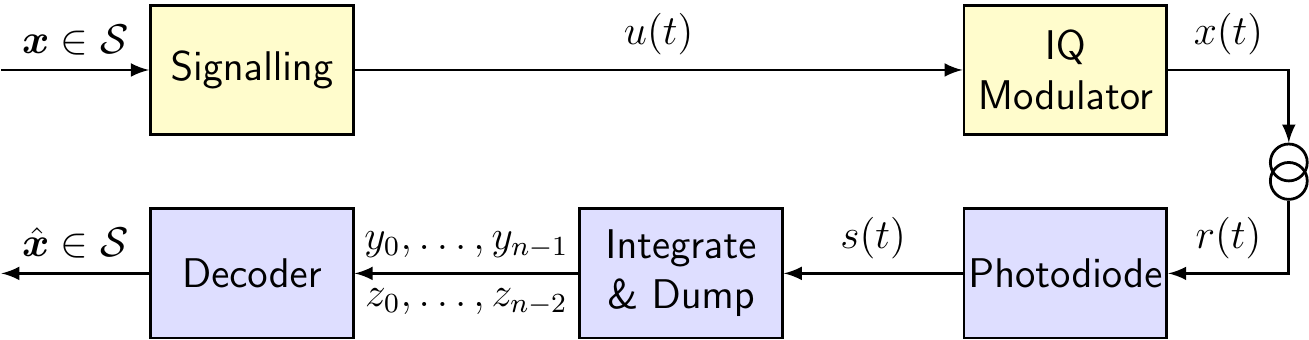}
\caption{The system model for operation in the O band.}
\label{fig:system_model_Oband}
\end{figure} 

\begin{figure}
\centering
\includegraphics[scale=0.6666666]{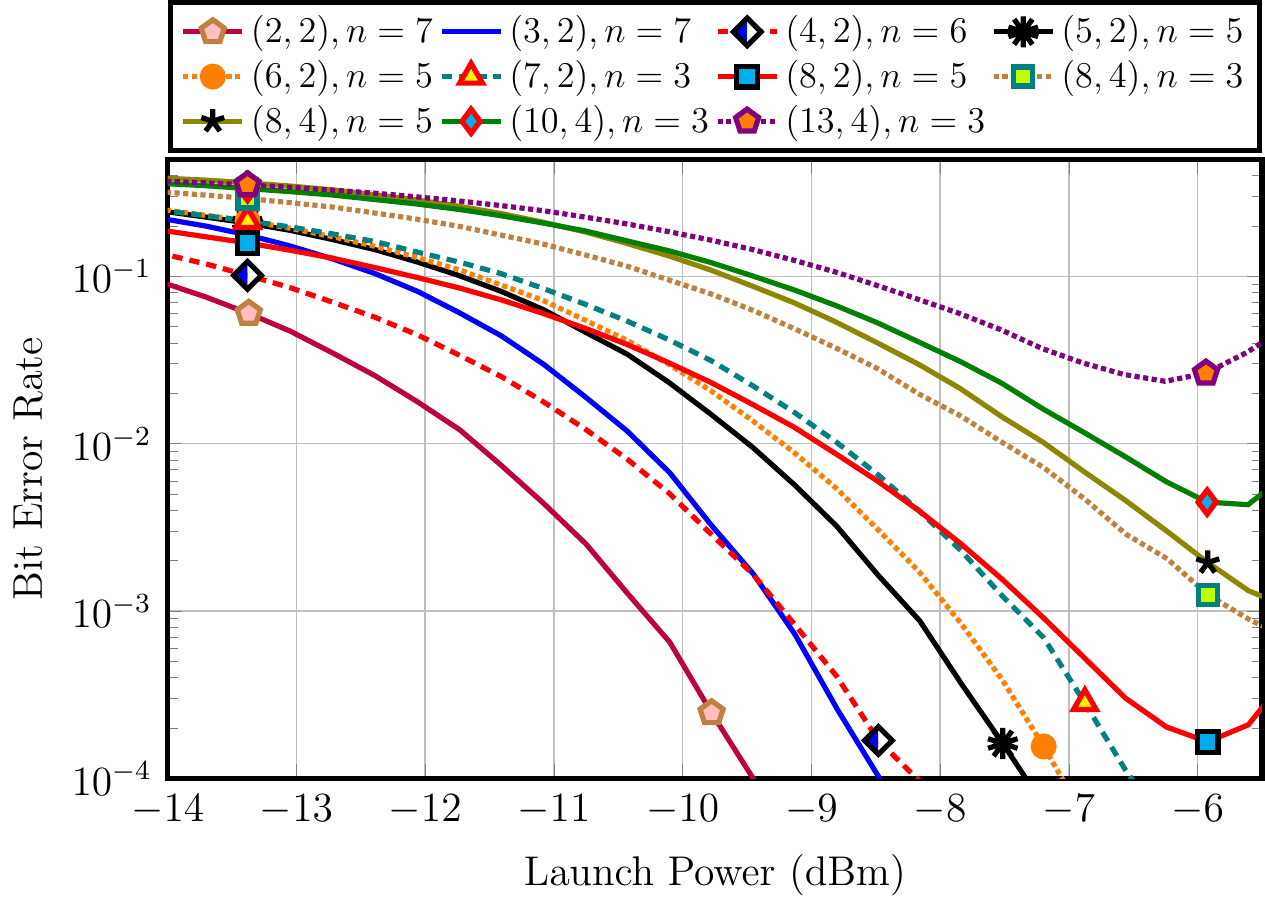}
\caption{The O-band BER  for the constellations highlighted in
Table~\ref{tab:ring_spacing}, at $50~$Gbaud, and at a laser power of
$2~\text{dBm}$. The chromatic dispersion in is
$1~\text{ps}\cdot\text{nm}^{-1}\cdot\text{km}^{-1}$.  Simulations
performed with a chromatic dispersion of
$-1~\text{ps}\cdot\text{nm}^{-1}\cdot\text{km}^{-1}$ gave
identical results.}
\label{fig:Oband-ber}
\end{figure}

\end{section}

\begin{section}{Discussion}
\label{sec:discussion}

In contrast to a long-haul or wireless communication system where the
power budget of the transmitter and the receiver are decoupled, in
short-haul applications the power budget is shared between the
transmitter and receiver.  As KK schemes use an unsynchronized laser
at the receiver---or send a tone along with the signal at
transmitter--- a fair comparison between KK and direct detection
under Tukey signalling or with IMDD should be made under a fixed
total power, \textit{i.e.}, including the receiver laser power (or
the transmitter tone), as well.  As a result, due to the high laser
power at the receiver (or high tone power at the transmitter), the KK
schemes are not power efficient for communication over short
distances, \textit{i.e.}, $<10~$km.  However, for longer distances,
\textit{e.g.}, about $100~$km, other imperfections of optical fiber
become substantial, necessitating post-detection compensation. In
those cases, by recovery of the complex-valued received waveform from
its intensity, KK receivers are a matter of interest.

In this paper, we addressed three practical issues which were ignored
in~\cite{tukey}. First, we addressed constellation design and
decoding complexity by introducing trellis diagrams for SQAM
constellations.  Second, have included the nonlinearity of the IQ
modulator, rather than assuming arbitrary waveform generation.
Third, we searched over various $n_r$, $n_p$, and $\delta$ values to
determine which constellations have best performance.

Comparisons with  IMDD show that at $50~$Gbaud and at a launch power
of $-10~$dBm, the proposed scheme achieves a throughput of $200~$Gb/s
using $(8,4)$-SQAM constellation, while at this launch power IMDD
achieves $145~$Gb/s using PAM-$8$.  This increase in the throughput
requires implementation of an IQ modulator at the transmitter and two
ADCs, each operating at the symbol rate, at the receiver.

\end{section}

\section*{Acknowledgment}
The authors would like to thank Prof. Anthony Chan Carusone, University of
Toronto, for helpful discussions.

\bibliographystyle{IEEEtran}
\bibliography{IEEEabrv,references}
\end{document}